\documentclass[a4paper,10pt]{article}
\usepackage{a4wide}
\usepackage[utf8]{inputenc}
\usepackage[T1]{fontenc}

\usepackage{subcaption}
\usepackage{amsmath,amssymb}
\usepackage{amsthm}
\usepackage{enumerate} 
\usepackage{authblk}

\usepackage[round]{natbib}
\bibpunct[, ]{(}{)}{;}{a}{}{,}

\usepackage[svgnames]{xcolor}
\usepackage[pdftex]{hyperref}
\hypersetup{
    colorlinks=true,
    colorlinks,
    linkcolor=Navy,
    citecolor=Navy,
    urlcolor=Navy
}

\definecolor{highlightNEW}{named}{black}

\usepackage{graphicx}
\graphicspath{{fig/}}
\usepackage{epstopdf}
\usepackage{booktabs}
\usepackage{multirow}
\usepackage{float}

\newtheorem{theorem}{Theorem}[section] 
\newtheorem{example}[theorem]{Example} 

\newtheorem{lemma}[theorem]{Lemma} 
\newtheorem{remark}[theorem]{Remark}

\newcommand{\doi}[1]{DOI~\href{\detokenize{http://dx.doi.org/#1}}{\detokenize{#1}}}

\newdimen\CdotAxis
\newcommand*{\CdotAux}[3]{%
  {%
    \settoheight\CdotAxis{$#2\vcenter{}$}%
    \sbox0{%
      \raisebox\CdotAxis{%
        \scalebox{#1}{%
          \raisebox{-\CdotAxis}{%
            $\mathsurround=0pt #2#3$%
          }%
        }%
      }%
    }%
    \dp0=0pt %
    \sbox2{$#2\bullet$}%
    \ifdim\ht2<\ht0 %
      \ht0=\ht2 %
    \fi
    \sbox2{$\mathsurround=0pt #2#3$}%
    \hbox to \wd2{\hss\usebox{0}\hss}%
  }%
}
\makeatletter
\def\mathcolor#1#{\@mathcolor{#1}}
\def\@mathcolor#1#2#3{%
  \protect\leavevmode
  \begingroup
    \color#1{#2}#3%
  \endgroup
}
\makeatother

\let\oldalpha\alpha
\renewcommand{\alpha}{\mathcolor{highlightNEW}{\oldalpha}}

\newcommand{\email}[1]{\href{mailto:#1}{#1}}

\title{\textcolor{Navy}{\textsc{Watanabe's expansion: A Solution for the convexity conundrum}}}

\author[1,2]{David Garcia-Lorite\thanks{Corresponding author, \email{david.garcia.lorite@gmail.com}}}
\author[3]{Ra\'{u}l Merino}

\affil[1]{CaixaBank, Quantitative Analyst Team, Plaza de Castilla, 3, 28046 Madrid, Spain,}
\affil[2]{Facultat de Matem\`{a}tiques i Inform\`{a}tica, Universitat de Barcelona, \authorcr Gran Via 585, 08007 Barcelona, Spain,\vspace*{3pt}}
\affil[3]{VidaCaixa S.A., Market Risk Management Unit, \authorcr C/Juan Gris, 2-8, 08014 Barcelona, Spain.}

\date{\normalfont\small\today}

\begin{document}

\maketitle
\begin{abstract}
In this paper, we present a new method for pricing CMS derivatives. We use Mallaivin's calculus to establish a model-free connection between the price of a CMS derivative and a quadratic payoff. Then, we apply Watanabe's expansions to quadratic payoffs case under local and stochastic local volatility. Our approximations are generic. To evaluate their accuracy, we will compare the approximations numerically under the normal SABR model against the market standards: Hagan's approximation, and a Monte Carlo simulation.
\end{abstract}

\section{Introduction}
Option pricing is one of the main tasks of a derivatives desk. Not only to obtain the derivatives prices and trade them but also to calibrate more complex models or manage all the associated risks. In practice, most institutions consider stochastic volatility models to better capture market dynamics. However, including this feature adds more complexity to the calculations, making the processes more complicated and requiring more time to execute.

Over the last decades, a research line has consisted of finding different methods to approximate the price accurately. On the one hand, there are approximations focused on finding an expansion of the implied volatility. For example, using perturbation methods as in \cite{HaganWoodward99}, \cite{Hagan02} and \cite{Jacquier14}, using heat kernel expansions as in \cite{Gatheral12} or \cite{Labordere08}, or asymptotic expansion method as in \cite{KaramiShiraya}. On the other hand, others focused on prices. For example, the Al\`os decomposition formula in \cite{Alos12} expanded later in \cite{Alos20} and \cite{GLMV19a}, adjoint expansions as in \cite{Pagliarani}, the cos methods as in \cite{Fang08} or Watanabe's expansion in \cite{Kunitomo03}, \cite{Takahashi} and \cite{Nagami21}.


Despite the vast literature on option approximations, there are few written references on quadratic payoffs. These products are not commonly traded but are useful to calculate the convexity's correction for CMS. The main references are \cite{Hagan19} and \cite{HaganWoodward20}. The value of a quadratic call, put, and swaps are given by
\begin{eqnarray}\label{quadratic_payoff_definition}
V^{QC}(T, F, K) :&=& N_{T_p} \mathbb{E}^{\mathbb{Q}^{N}}\left[\left(F_{T}- K\right)^{2}_{+}\right],\\
V^{QP}(T, F, K) :&=& N_{T_p} \mathbb{E}^{\mathbb{Q}^{N}}\left[\left(K-F_{T}\right)^{2}_{+}\right],\\
V^{QS}(T, F, K) :&=& N_{T_p} \mathbb{E}^{\mathbb{Q}^{N}}\left[\left(F_{T}- K\right)^{2}\right].
\end{eqnarray}
where $T$ is the expiration date, $T_p$ the payment date and $F_T$ is the forward rate process associated to the derivative. On the other hand, $\mathbb{E}^{\mathbb{Q}^{N}}\left[\cdot\right]$ is the expected value associated to the numeraire $N_t$. Although there is a closed formula for normal and log-normal when volatility is constant, it is not good enough since these products depend more on market skews and smiles than vanilla options. 
Many banks use replication to price these structures. Note that
\begin{eqnarray}
\left(F_{T}- K\right)^{2}_{+} = 2 \int^{\infty}_{K} \left(F_{T}- K\right)_{+} dK',
\end{eqnarray}
therefore, it is possible to approximate the value as the sum of options with different strikes. That is the main advantage of this approach, as it is consistent with vanilla option pricing, easy to implement, and automatically satisfies the put-call parity. Unfortunately, from a practical standpoint, it is necessary a dense option grid with different strikes. Only a few options are traded per expiration and strike, so it is essential to calculate the grid, making it computationally intensive. In addition, the option prices with higher strikes, which are required for the calculation, are illiquid and have a `fictitious' market price. To overcome this, an alternative methodology based on finding an approximation formula under the SABR model has been proposed by \cite{Hagan19} and \cite{HaganWoodward20}.

The paper aims to obtain an alternative approach for pricing CMS derivatives. Firstly, we will establish a connection between the price of a CMS derivative and a quadratic payoff using a bit of Mallivin calculus. This relationship is obtained without relying on any specific model. Then, we will extend the results of Watanabe's expansions to the case of quadratic payoffs over local and stochastic local volatility. The approximations are generic. To evaluate their accuracy, we will compare the approximations numerically under the normal SABR model against the market standards: Hagan's approximation, and a Monte Carlo simulation.

The structure of the paper is as follows. In Section \ref{sec:Intro_CMS}, we will introduce the main concepts of the CMS market to make the paper self-contained. In Section \ref{sec:CMSConvexity}, we explain how to derive a general expression for the convexity adjustment for CMS caps, CMS floors, and CMS swaps using Malliavin calculus. This expression will depend on the different quadratic payoffs explained above. In Section \ref{sec:Watanabe}, we provide a brief explanation of the Watanabe calculus. We will also give an example of the normal SABR model contrasting the accuracy of the approximation. In Section \ref{sec:QuadraticOptionPrice}, we prove how to use Watanabe calculus to obtain the price of a quadratic call, quadratic put, and quadratic swaps under different models. The approximation is applied to local and stochastic local volatility models. In particular, the accuracy is tested for the normal SABR model by comparing the results against Monte Carlo and Hagan's approximation. Finally, in Section \ref{sec:Conclusion}, we will present our conclusions. 

We provide the code for running all the examples on the papper at \url{https://github.com/Dagalon/PyStochasticVolatility}.

\section{A brief introduction to the CMS market}\label{sec:Intro_CMS}
One of the principal uses of quadratic payoffs is the convexity adjustment calculation, especially for CMS products. Therefore, to make the article self-contained, in this section, we will do a brief introduction to the CMS market.

We will consider the market standard of two curves: one to discount the cash flows and another to estimate the floating coupon. These two curves are equivalent to the discount and repo curves in the equity market. \\

The following notation will be used throughout the paper:
\begin{itemize}
		\item We will denote the discount curve by $d$ and the estimated curve by $e$.
    \item $P^{i}(T)$ is the spot discount factor to T using the curve i=d, e. 
    \item $P^{i}(t,T)= \frac{P^{i}(T)}{P^{i}(t)}$ is the forward discount factor associated to the curve i=d, e from t to T.
		\item $\Delta(T_a,T_b)$ is the fraction's year between $T_a$ and $T_b$.
    \item $L^{i}(t,T_a,T_b) = \frac{1}{\Delta(T_a,T_b)}\Big(\frac{P^{i}(t,T_a)}{P^{i}(t,T_b)} - 1\Big)$ is the floating rate associated with the curve i=d, e.
\end{itemize}

The swap is the most commonly traded derivative in the interest rate market. While there are various swaps types, the market standard is to exchange a fixed coupon for a floating one until the end of the contract. Fixed and floating rate payments can be scheduled on different dates and periods. But, to simplify the notation, we assume that payments are made on the same dates and in the same fractions of the year. Today is denoted as $T_{0}$, and the future payment dates are scheduled by 
\begin{eqnarray*}
\mathcal{T}=\left\{T_{a}, T_{a+1}, \ldots, T_{b}\right\}.
\end{eqnarray*}

The standard swap price is 
\begin{align}\label{swap_payoff}
V^{Swap}(T_0,T_{a}, T_{b}, C) &= w * \Biggl(\sum_{i=a}^{b} \Delta(T_{i-1},T_{i})P^{d}(T_0,T_{i}) C  \nonumber \\
& \qquad-    \sum_{i=a}^{b} \Delta(T_{i-1},T_{i})P^{d}(T_0,T_{i}) \mathbb{E}_{T_0}^{T_i}\bigg[L^{e}(T_{i-1},T_{i-1},T_{i})\bigg]\Biggl).
\end{align}


When we receive a fixed coupon, $w=1$, we are in a receiver swap. On the other hand, if we pay the fixed coupon, $w=-1$, we are in a payer swap. The measure $\mathbb{Q}^{T_i}$ used on the floating leg is associate to the numeraire $N^{T_p}_t:= P^{d}(t, T_p)$. 

There are two key concepts when working with swaps: the swap rate and the annuity factor. The swap rate is the fixed coupon that makes the structure fair. It is defined by 
\begin{align*}
S_{a,b}(T_0) &= \frac{\sum_{i=a}^{b} \Delta(T_{i-1},T_{i})P^{d}(T_0,T_{i}) \mathbb{E}_{T_0}^{T_i}\bigg[L^{e}(T_{i-1},T_{i-1},T_{i})\bigg]}{\sum_{i=a}^{b} \Delta(T_{i-1},T_{i})P^{d}(T_0,T_{i})}.
\end{align*}
The annuity factor is defined by 
\begin{align*}
01(T_0, T_a, T_b) &= \sum_{i=a}^{b} \Delta(T_{i-1},T_{i})P^{d}(T_0,T_{i}).
\end{align*}

A constant maturity swap (CMS) is a fixed-for-floating swap, where the floating leg pays, or receives, periodically a swap rate with a fixed time to maturity. It has the same structure as the standard swap, but changing the floating payment from $\mathbb{E}^{T_i}\big[L^{e}(T_{i-1},T_{i-1},T_{i})\big]$ to $\mathbb{E}_{T_0}^{T_p}\big[S_{i-1,i-1+b}(T_{i-1})\big]$ with $T_{i-1} \leq T_{p_i} < T_i$.\\

The CMS swap price is given by
\begin{align}\label{cms_swap_payoff}
V^{Swap}_{CMS}(T_0,T_{a}, T_{b}, C) &= w * \Biggl(\sum_{i=a}^{b} \Delta(T_{i-1},T_{i})P^{d}(T_0,T_{i}) C \nonumber \\
&\qquad - \sum_{i=a}^{b} \Delta(T_{i-1},T_{i})P^{d}(T_0,T_{p_i}) \mathbb{E}_{T_0}^{T_{p_i}}\bigg[S_{i-1,i-1+b}(T_{i-1})\bigg] \Biggl).
\end{align}

There is a significant difference between valuing a standard swap and a CMS swap. The floating coupon of the standard swap is a martingale under the forward measure $N^{T_p}_t$. In other words,
\begin{equation*}
    \mathbb{E}_{T_0}^{T_i}\bigg[L^{e}(T_{i-1},T_{i-1},T_{i})\bigg] = L^{e}(T_0,T_{i-1},T_{i}).
\end{equation*}
But the swap rate is not, it implies that
\begin{equation*}
    \mathbb{E}_{T_0}^{T_p}\bigg[S_{i-1,i-1+b}(T_{i-1})\bigg] \neq S_{i-1,i-1+b}(T_0).
\end{equation*}
The swap rate is a martingale respect to the measure induced by the numeraire $N^{a,b}_t:=01(t, T_a, T_b)$. Since the floating rate and the swap rate are martingales with respect to different numeraries, a convexity adjustment is necessary for CMS products.

In fact, we have that
\begin{equation*}
    \mathbb{E}_{T_0}^{T_{p_i}}\bigg[S_{i-1,i-1+b}(T_{i-1})\bigg] = \frac{N^{i-1,i-1+b}_{T_{0}}}{P^{d}(T_{0}, T_{p_i})} \mathbb{E}_{T_0}^{T_{p_i}}\bigg[S_{i-1,i-1+b}(T_{i-1}) \frac{N^{i-1,i-1+b}_{T_{i-1}}}{P^{d}(T_{i-1}, T_{p_i})}   \bigg].
\end{equation*}

The essence of CMS modeling lies on the definition of a mapping function $M(t,a,b)$. In \cite{AndreasenPiterbargIII}, a comprehensive discussion is presented regarding the selection of the mapping function and the requirements it must satisfy. We use the following mapping function:
\begin{equation*}
M(t,a,b,p)=\frac{P^{d}(t,T_p)}{N_t^{a,b}}.
\end{equation*}
Then, we have that
\begin{align}\label{convexity_adjustment_cms_swap}
    \mathbb{E}_{T_0}^{T_{p_i}}\bigg[S_{i-1,i-1+b}(T_{i-1})\bigg] &=  \mathbb{E}_{T_0}^{N^{i-1,i-1+b}}\bigg[S_{i-1,i-1+b}(T_{i-1}) \frac{M(T_{i-1},i-1,i-1+b, p_i)}{M(T_{0},i-1,i-1+b, p_i)}  \bigg] \nonumber \\
    &= S_{i-1,i-1+b}(T_{0}) \nonumber \\
		&+ \mathbb{E}_{T_0}^{N^{i-1,i-1+b}}\bigg[S_{i-1,i-1+b}(T_{i-1}) \left(\frac{M(T_{i-1},i-1,i-1+b, p_i)}{M(T_{0},i-1,i-1+b, p_i)} - 1 \right)  \bigg].
\end{align}
The second term in the summand of the right side is known as convexity adjustment for the CMS rate, see \cite{Hagan03} for more details.\\

Concluding this section, we introduce CMS Caps and Floors, the most common CMS derivative products. The paper aims to develop different price expansions under various models for calculating their convexity adjustments. A CMS Cap and CMS Floor consist of a series of calls and puts, named caplets or floorlets, with the underlying $S_{a,b}(t)$. The payoff is given by
\begin{align}\label{cms_cap_floor_payoff}
V^{Cap}_{CMS}(T_0,T_a,T_b, K) &= \left(\sum_{i=a}^{b} \Delta(T_{i-1},T_{i})P^{d}(T_0,T_{p_i}) \mathbb{E}_{T_0}^{T_{p_i}}\bigg[S_{i-1,i-1+b}(T_{i-1})-K)^{+}\bigg] \right), \nonumber \\
V^{Floor}_{CMS}(T_0,T_a,T_b, K) &= \left(\sum_{i=a}^{b} \Delta(T_{i-1},T_{i})P^{d}(T_0,T_{p_i}) \mathbb{E}_{T_0}^{T_{p_i}}\bigg[K-S_{i-1,i-1+b}(T_{i-1}))^{+}\bigg] \right). \nonumber 
\end{align}
Now, we will focus on calculating the value of a single caplet, that is
\begin{equation*}
P^{d}(0,T_p) \mathbb{E}^{T_p}\bigg[ \Big(S_{a,b}(T_a) - K\Big)^{+}\bigg] = P^{d}(0,T_a) \mathbb{E}^{T_a}\bigg[ \Big(S_{a,b}(T_a) - K\Big)^{+} P^{d}(T_a,T_p)\bigg].
\end{equation*}
Using a change of numeraire, we have that 
\begin{align*}
N^{a,b}_0 \mathbb{E}^{N^{a,b}}\bigg[ \Big(S_{a,b}(T_a) - K\Big)^{+} M(T_a,a,b,p)\bigg] &= N^{a,b}_0 \mathbb{E}^{N^{a,b}}\bigg[ \Big(S_{a,b}(T_a) - K\Big)^{+} \bigg]\nonumber \\ 
&+ N^{a,b}_0  \mathbb{E}^{N^{a,b}}\bigg[ \Big(S_{a,b}(T_a) - K\Big)^{+} \Big(M(T_a,a,b,p) - 1\Big) \bigg] 
\end{align*}
where 
$N^{a,b}_0 = 01(0,a,b)$
and $M(t,a,b,p)=\frac{P(t,T_p)}{01(t,a,b)}.$ The second term is the convexity adjustment for a CMS Cap. To calculate this term, it is usually assumed that
\begin{equation*}
M(t,a,b,p)= M_{p}\big(S_{a,b}(t)\big),
\end{equation*}
and use the following replication formula
\begin{equation*}
\bigg(g\Big(S_{a,b}(t)\Big)\bigg)^{+}=g^{\prime}(K)\Big(S_{a,b}(t) - K\Big)^{+} + \int_{K}^{\infty} \Big(S_{a,b}(t) - u\Big)^{+} g^{\prime \prime}(u) du
\end{equation*}
with 
$$g\Big(S_{a,b}(t)\Big)= \Big(S_{a,b}(t) - k\Big) \Big(M_{p}(S_{a,b}(t))-1\Big).$$

The primary concern with this formula is that it requires knowledge of the swaption price on the strike grid where the right side tends to $\infty$. Consequently, the convexity adjustment heavily relies on the extrapolation assumption.

In recent papers such as \cite{Hagan19} and \cite{HaganWoodward20}, the authors have proposed a new methodology that expresses the convexity adjustment for CMS Caps, CMS Floors, and CMS Swaps as a quadratic payoff. Using the singular perturbation technique, they obtain an expansion of the implied volatility for the quadratic payoff when the dynamics of the underlying follow a SABR-type model. 

In the following sections, we will demonstrate how to use the Watanabe calculus to compute the price of quadratic payoff under a general stochastic local volatility model.

\section{The connection between Quadratic Payoffs and the CMS Cap-Floor convexity}\label{sec:CMSConvexity}
In this section, we will illustrate the connection between the quadratic payoffs and the convexity adjustment for CMS Cap, Floor, and Swap instruments. We will use Malliavin's techniques to derive a generic representation for this adjustment. For more detailed content on Malliavin calculus see \cite{AlosLorite} or \cite{Nualart}. 

We will start with the convexity adjustment for a CMS Cap, as we saw in Section \ref{sec:Intro_CMS}. We have that
\begin{align}\label{ca_cms_cap}
CA_{CMS}^{cap}(0,a,b)&= N^{a,b}_0 \mathbb{E}^{N^{a,b}}\Bigg[ \Big(S_{a,b}(T_a) - K\Big)^{+} \bigg(M(T_a,a,b,p) - 1\bigg) \Bigg] \nonumber \\
&= N^{a,b}_0 \mathbb{E}^{N^{a,b}}\Bigg[ \Big(S_{a,b}(T_a) - K\Big)^{+} \bigg( \mathbb{E}^{N^{a,b}}\Big[M(T_a,a,b,p)|S_{a,b}(T_a)\Big] - 1\bigg) \Bigg] \nonumber \\
&= N^{a,b}_0 \mathbb{E}^{N^{a,b}}\Bigg[ \Big(S_{a,b}(T_a) - K\Big)^{+} \bigg(M(S_{a,b}(T_a)) - 1\bigg) \Bigg] \nonumber \\
\end{align}
where we will define 
\begin{equation*}
M\big(S_{a,b}(T_a)\big):= \mathbb{E}^{N^{a,b}}\Big[M(T_a,a,b,p)|S_{a,b}(T_a)\Big].
\end{equation*}
Here, the function $M\big(S_{a,b}(T_a)\big)$ is called the mapping function. A popular choice for this function is the linear model. In this model, it is assumed that
\begin{equation*}
M\big(S_{a,b}(T_a)\big)= \alpha(T_a) S_{a,b}(T_a) + \beta(T_a).
\end{equation*}
The parameters $\alpha(T_a)$ and $\beta(T_a)$ are chosen according to the properties of $S_{a,b}(T_a)$ under the measure associated with the numeraire $N_t^{a,b}$. Using the Clark-Ocone formula with $F_t = S_{a,b}(t)$, we have
\begin{equation*}
M_{S_{a,b}(T_a)}= M_{S_{a,b}(0)} + \int_{0}^{T_a} \mathbb{E}_s^{N^{a,b}}\bigg[ \partial_{S} M\big(S_{a,b}(T_{a})\big) D^{W}_s S_{a,b}(T_{a})  \bigg] dW_s.
\end{equation*}
If we freeze $\partial_{S} M\big(S_{a,b}(T_{a})\big)$ in $S_{a,b}(0)$, we get the following approximation
\begin{align*}
M_{S_{a,b}(T_a)} &\approx M_{S_{a,b}(0)} + \partial_{S} M\big(S_{a,b} (0)\big) \int_{0}^{T_a} \mathbb{E}_s^{N^{a,b}}\bigg[  D^{W}_s S_{a,b}(T_{a}) \bigg] dW_s \\
&= M\big(S_{a,b}(0)\big) + \partial_{S} M\big(S_{a,b}(0)\big) \Big(S_{a,b}(T_a) - S_{a,b}(0)\Big).
\end{align*}
Therefore, if we substitute the above equality in \eqref{ca_cms_cap}, we get that
\begin{equation*}
CA_{CMS}^{cap}(0,a,b) \approx \Big(M\big(S_{a,b}(0)\big) - 1\Big) V^{C}(T_a,S_{a,b},K) + \Big(\partial_{S} M\big(S_{a,b}(0)\big)-1\Big) V^{QC}(T_a,S_{a,b},K).
\end{equation*}
Likewise, in the CMS Floor case, we obtain that
\begin{equation*}
CA_{CMS}^{floor}(0,a,b) \approx \Big(M\big(S_{a,b}(0)\big) - 1\Big) V^{P}(T_a,S_{a,b},K) + \Big(\partial_{S} M\big(S_{a,b}(0)\big)-1\Big) V^{QP}(T_a,S_{a,b},K)
\end{equation*}
and, in the swap case, we find that
\begin{equation*}
CA_{CMS}^{swap}(0,a,b) \approx \Big(M\big(S_{a,b}(0)\big) - 1\Big) \Big(S_{a,b}(0) - K\Big) + \Big(\partial_{S} M\big(S_{a,b}(0)\big)-1\Big) V^{QS}(T_a,S_{a,b},K).
\end{equation*}

\begin{example}
In the SABR model, to compute $P^{d}(0,T_p)\mathbb{E}^{T_p}\Big[S_{a,b}(T_a)\Big]$ when $T_a < T_p < T_b $, practitioners often use the replication formula
\begin{align*}
g(S_{a,b}(T_a)) &= g(S_{a,b}(T_a)) + \Big(S_{a,b}(T_a) - S_{a,b}(0)\Big)g^{\prime}(S_{a,b}(0))  \\
 &+ \int_{-\infty}^{S_{a,b}(0)} g^{\prime \prime}(K) \Big(K-S_{a,b}(T_a)\Big)^{+} dK + \int_{S_{a,b}(0)}^{\infty} g^{\prime \prime}(K) \Big(S_{a,b}(T_a)-K\Big)^{+} dK. 
\end{align*} 
Then, if we take $\mathbb{E}^{N^{a,b}}\big[\cdot\big]$, we get that
\begin{align}\label{replication_formula}
N^{a,b}_{0}\mathbb{E}^{N^{a,b}}\Big[g(S_{a,b}(T_a))\Big] &= N^{a,b}_{0} g(S_{a,b}(0)) + N^{a,b}_{0} \int_{-\infty}^{S_{a,b}(0)} g^{\prime \prime}(K) \mathbb{E}^{N^{a,b}}\Big[(K-S_{a,b}(T_a))^{+}\Big] dK  \nonumber \\
 &+ N^{a,b}_{0} \int_{S_{a,b}(0)}^{\infty} g^{\prime \prime}(K) \mathbb{E}^{N^{a,b}}\Big[(S_{a,b}(T_a)-K)^{+}\Big] dK.
\end{align}
Now, we have that
\begin{equation*}
P^{d}(0,T_p)\mathbb{E}^{T_p}\Big[S_{a,b}(T_a)\Big] = N^{a,b}_{0} \mathbb{E}^{N^{a,b}}\Big[M(S_{a,b}(T_a)) S_{a,b}(T_a)\Big].
\end{equation*}
Then, we can utilize the replication formula \eqref{replication_formula} with $g(S_{a,b}(T_a)) = M(S_{a,b}(T_a)) S_{a,b}$ to calculate the convexity adjustment. However, this approach requires assumptions about the distributional behavior of $S_{a,b}(T_a)$ in both the left and right tails. It is necessary to choose a cut-off point until the SABR parametrization works accurately and implement an extrapolation mechanism outside this region. 

There is another approximation using the quadratic payoff provided by the following formula:
\begin{align*}
N^{a,b}_{0} \mathbb{E}^{N^{a,b}}\Big[M(S_{a,b}(T_a)) S_{a,b}\Big] &\approx S_{a,b}(0) + N_0^{a,b} \Big(\partial_{S} M(S_{a,b}(0))-1\Big) \mathbb{E}^{N^{a,b}}\Big[ S^2_{a,b}(T_a)\Big] \\
& = S_{a,b}(0) + N_0^{a,b} \Big(\partial_{S} M(S_{a,b}(0))-1\Big) S^{2}_{a,b}(0) \\ 
&+\Big(\partial_{S} M(S_{a,b}(0))-1\Big) V^{QC}(T_a,S_{a,b}(0),S_{a,b}(0)).
\end{align*} 
The essence of this approach is to calculate, or approximate, the value of $V^{QC}(T_a,S_{a,b}(0),S_{a,b}(0))$.
\end{example}

\section{Basic introduction to Watanabe's exapansion}\label{sec:Watanabe}
In this section, we will give a brief introduction to Watanabe's expansion in the field of quantitative finance, see \cite{Kahl} for more details. This method supposes that we can find an expansion of the underlying $F_t$ in terms of a small parameter $\epsilon << 1$
\begin{align}\label{base_expansion}
    	F_{\epsilon}(T) &= F_{0}+ \epsilon \cdot g_{1}(T) + \epsilon^{2} \cdot g_{2}(T) + \epsilon^{3} \cdot g_{3}(T) + \cdots \nonumber  \\
    	 &= F_{0} + \epsilon \sigma_{0}\sqrt{T} \Big(\hat{g}_{1}(T) + \epsilon \hat{g}_{2}(T) + \epsilon^{2} \hat{g}_{3}(T) + \cdots  \Big)
\end{align}
where 
\begin{eqnarray*}
\hat{g}_{i}(T) = \frac{g_{i}(T)}{\sigma_{0}\sqrt{T}}.
\end{eqnarray*}
The expansion is constructed such that $\hat{g}_{1} \sim \mathcal{N}(0,1)$. We refer to $\sigma_{0}$ as the initial volatility and $T$ as the time to maturity. We use the auxiliary term $\epsilon$ to derive the asymptotic option pricing formula.
To calculate 
\begin{eqnarray*}
V^{C}(T,F_T,K):=\mathbb{E}^{\mathbb{Q}_N}\Big[(F_T - K)^{+}\Big]
\end{eqnarray*}
we will use the forward expansion \eqref{base_expansion}. It is easy to show that
\begin{align}\label{expansion_base_generic_payoff}
\frac{V^{C}(T,F,K)}{N_T} &= \mathbb{E}^{\mathbb{Q}_N}\Bigg[\bigg(x_{0} - K + \epsilon \sigma_{0} \sqrt{T} \sum_{j=1}^{\infty} \epsilon^{j-1} \hat{g}_{j}(T)\bigg)^{+} \Bigg] \nonumber  \\
&= \epsilon \sigma_{0} \sqrt{T} \mathbb{E}^{\mathbb{Q}_N}\Bigg[\bigg(\hat{g}_{1}(T) - y + \epsilon{}\sum_{j=1}^{\infty} \epsilon^{j-1}\hat{g}_{j+1}(T) \bigg)^{+} \Bigg] 
\end{align}
where 
\begin{eqnarray}\label{label_y}
y = \frac{K - F_{0}}{\epsilon \sigma_{0} \sqrt{T}}
\end{eqnarray}
and
\begin{eqnarray*}
V^{C}_{(j)} := \partial^{(j)}_{F} (F - K)^{+}_{|F=\hat{g}_{1}(T),K=y}.
\end{eqnarray*}
Then, using a Taylor's expansion on $V^{C}(T,T,T,T,\cdot,K)$ around $\hat{g}_{1}(T) - y$, we get the following price expansion
\begin{align}\label{expansion_generic_payoff_order_5}
\frac{V^{C}(T,F,K)}{N_T} &= \epsilon \sigma_{0} \sqrt{T} \mathbb{E}^{\mathbb{Q}^{N}}\Bigg[\Big(\hat{g}_{1}(T) - y\Big)^{+} +  \epsilon  V^{C}_{(1)} \hat{g}_{2}(T)  \nonumber \\
& \left.+ \epsilon^{2} \bigg(V^{C}_{(1)} \hat{g}_{3}(T) + \frac{1}{2} V^{C}_{(2)} \hat{g}_{2}^{2} \bigg) \right.  \nonumber \\
& \left. +\epsilon^{3}  \bigg( V^{C}_{(1)}\hat{g}_{4} +V^{C}_{(2)}\hat{g}_{2}(T) \hat{g}_{3}(T) + \frac{1}{6} V^{C}_{(3)}\hat{g}^{2}_{2}(T) \bigg) \right. \nonumber \\ 
&+ \epsilon^{4}  \bigg(V^{C}_{(1)} \hat{g}_{5}(T) + V^{C}_{(2)}\Big(\hat{g}_{2}(T)\hat{g}_{4}(T) + \frac{1}{2}\hat{g}^{2}_{3}(T) \Big) \nonumber \\
&\hspace{1cm}+ \frac{1}{2} V^{C}_{(3)} \hat{g}^{2}_2(T) \hat{g}_3(T) + \frac{1}{24} V^{C}_{(4)} \hat{g}^{2}_{2}(T) \bigg) + O(\epsilon^{5})\Bigg].
\end{align}

Note that to use the price expansion given in \eqref{expansion_generic_payoff_order_5}, we need to know the terms $\hat{g}_i(T)$ for $i=1, \cdots, n$. These terms  are expressed as 
\begin{equation}\label{g_i_definition}
\hat{g}_i(T) = \beta_i(T) I_{(\underbrace{0,\cdots,0}_\text{n-times},\underbrace{1,\cdots,1}_\text{$m_B$-times},\underbrace{1,\cdots,1}_\text{$m_W$-times})}(T)
\end{equation}
where $\beta_i(T)$ is a non-stochastic function and 
\begin{eqnarray*}
I_{(\underbrace{0,\cdots,0}_\text{n-times},\underbrace{1,\cdots,1}_\text{$m_B$-times},\underbrace{1,\cdots,1}_\text{$m_W$-times})}(T)
\end{eqnarray*}
is the next iterated integral

\begin{equation}\label{g_i_integral_representation}
\int_{0}^{T}\ \cdots \int_{0}^{t_{m_B}} \int_{0}^{t^{\prime}_1} \cdots \int_{0}^{t^{\prime}_{m_W}}    \int_{0}^{t^{\prime \prime}_1} \cdots \int_{0}^{t^{\prime \prime}_{n}}    du^{\prime \prime}_{n + 1} \cdots  du^{\prime \prime}_1  dB_{u^{\prime}_{m_B + 1}} \cdots  dB_{u^{\prime}_1}   dW_{u_{m_B+ 1}} \cdots  dW_{u_1}           
\end{equation} 
with $W_t$ and $B_t$ correlated Brownian motions of the underlying and volatility processes, respectively. The integral (\ref{g_i_integral_representation}) may seem complex, but using the Brownian bridge and conditioning on $\hat{g}_1(T)$, the expectation of the integral (\ref{g_i_definition}) is simplified to

\begin{equation}
\mathbb{E}^{\mathbb{Q}_N}\Bigg[ I_{(\underbrace{0,\cdots,0}_\text{n-times},\underbrace{1,\cdots,1}_\text{$m_B$-times},\underbrace{1,\cdots,1}_\text{$m_W$-times})}(T) \Bigg] = \int_{0}^{T} g_{(\underbrace{0,\cdots,0}_\text{n-times},\underbrace{1,\cdots,1}_\text{$m_B$-times},\underbrace{1,\cdots,1}_\text{$m_W$-times})}(u) \hat{\phi}_1(u) du
\end{equation}
where  $\hat{\phi}_1(u)$ is the pdf of  $\hat{g}_1(T)$ and  
\begin{eqnarray*}
g_{(\underbrace{0,\cdots,0}_\text{n-times},\underbrace{1,\cdots,1}_\text{$m_B$-times},\underbrace{1,\cdots,1}_\text{$m_W$-times})}(u) 
\end{eqnarray*}
is a function to be defined.\\

To illustrate, we will consider the normal SABR model and we will obtain the order $O(\epsilon^{2})$ approximation as described above.

\begin{remark}
To simplify the notation, when we have a single Brownian motion driving the dynamics of the underlying, then we will denote
\begin{equation*}
I_{(0,\cdots,0,1,\cdots,1)}(T) = I_{(\underbrace{0,\cdots,0}_\text{n-times};;\underbrace{1,\cdots,1}_\text{$m_W$-times})}(T).
\end{equation*}
\end{remark}

\begin{example}
The objective of this example is to provide a brief overview of how can be used the Watanabe calculus to approximate the price of a call option when the underlying follows a normal SABR model. Consider that the forward price dynamics are given by
\begin{align*}
dF_t &= \sigma_t dW_t \\
d\sigma_t &= \nu \sigma_t dB_t.
\end{align*}
with $<dW_t,dB_t>=\rho dt$. Then, if we will suppose that $\sigma_t$ is small, i.e there is a parameter $\epsilon$ such that $\sigma_t = \epsilon \sigma_{\epsilon,t}$. Therefore, we have that 
\begin{align}\label{forward_normal_sabr}
dF_{\epsilon, t} &= \epsilon \sigma_{\epsilon,t} dW_t \nonumber \\
d\sigma_{\epsilon,t} &= \nu \sigma_{\epsilon,t} dB_t.
\end{align}

\begin{lemma}
Under the dynamics given by \eqref{forward_normal_sabr}, we can re-write the forward price dynamics as:
\begin{align}
F_{\epsilon,T} = F_{0} &+\epsilon \sigma_{0} W_{T} + \nu \epsilon^{2} \sigma_{0}\int^{T}_{0} \int^{u}_{0}  dB_{u_{1}}  dW_{u} \nonumber\\
&+ \nu^{2} \epsilon^{3}\sigma_{0} \int^{T}_{0}\int^{u}_{0}\int^{u_{1}}_{0}  dB_{u_{2}}  dB_{u_{1}}  dW_{u} + \ldots\nonumber\\
&+ \nu^{n-1}\epsilon^{n}\sigma_{0} \int^{T}_{0}\int^{u}_{0}\int^{u_{1}}_{0} \cdots \int^{u_{n-2}}_{0} dB_{u_{n-1}} \cdots  dB_{u_{1}}  dW_{u} \nonumber\\
&+ \nu^{n}\epsilon^{n+1} \int^{T}_{0}\int^{u}_{0}\int^{u_{1}}_{0} \cdots \int^{u_{n-1}}_{0} \sigma_{\epsilon,u_{n}}dB_{u_{n}} \cdots  dB_{u_{1}}  dW_{u}.\nonumber
\end{align}
\end{lemma}
\begin{proof}
The result is straightforward by iteratively applying the It\^o-Taylor expansion.
\end{proof}

For practical reasons, an option price approximation is sought. Therefore, we consider an approximation of the forward price dynamics.
\begin{remark}\label{4th_fwd}
The 4th order forward price dynamics is
\begin{equation*}
F_{\epsilon,T} = F_{0} + \epsilon \sigma_{0} W_{T} + \epsilon^{2} \nu  \sigma_{0}\int^{T}_{0} \int^{u}_{0}  dB_{u_{1}}  dW_{u} + \epsilon^{3} \nu^{2}\sigma_{0} \int^{T}_{0}\int^{u}_{0}\int^{u_{1}}_{0}  dB_{u_{2}}  dB_{u_{1}}  dW_{u} + O(\epsilon^{4}).\nonumber
\end{equation*}
\end{remark}

To obtain the Watanabe price expansion \eqref{expansion_generic_payoff_order_5}, we need to find an expression for the $\hat{g}_{i}$ terms. Using Remark \ref{4th_fwd} makes it easy to see the analytical expression for each term. Therefore
\begin{align}\label{g_normal_sabr}
\hat{g}_{1}(T) &=\frac{W_{T}}{\sqrt{T}},\nonumber\\
\hat{g}_{2}(T) &= \frac{\nu}{\sqrt{T}} \int^{T}_{0} \int^{u}_{0}  dB_{u_{1}}  dW_{u},\nonumber\\
\hat{g}_{3}(T) &= \frac{\nu^{2}}{\sqrt{T}}  \int^{T}_{0}\int^{u}_{0}\int^{u_{1}}_{0}  dB_{u_{2}}  dB_{u_{1}}  dW_{u}.
\end{align}
On other hand, we have that
\begin{align}
V^{C}_{(1)}(F,K)&=I(F>K), \label{id_function} \\
V^{C}_{(2)}(F,K)&=\delta(F-K) \label{dirac_function}.
\end{align}

\begin{theorem}[Option price expansion of 2nd order under normal SABR]
Under the model \eqref{forward_normal_sabr}, the 2nd order call option approximation price is
\begin{align}\label{sabr_call_option_approximation}
\frac{V^{C}(T,F,K)}{N_T}   = \sigma_{0} \sqrt{T} \Bigg(G(y) &+ \frac{1}{2}  \rho \nu y \phi(y) \sqrt{T} +  \frac{1}{24}  \nu^{2} \bigg( \rho^{2} \Big(3\big(y^{2} - 1\big)^{2} + 4y^{2} + 2\Big)  \bigg) \phi(y)T \nonumber 
\\&+ \frac{2}{3}\hat{\rho}^{2} \Big(y^{2}+2\Big) - 1 \Bigg) + O(\epsilon^3)
\end{align}
where $y$ is given by \eqref{label_y}, $G(y)= \phi(y) - y \bar{\Phi}(y)$, $\phi(\cdot)$ is standard normal pdf, $\Phi(\cdot)$ the standard normal cdf and $\bar{\Phi}(\cdot)=1-\Phi(\cdot)$.
\end{theorem}
\begin{proof}
We consider the price expansion \eqref{expansion_generic_payoff_order_5} until order $\epsilon^{2}$. Then, we use the calculation of the $\hat{g}_{i}$ terms given in \eqref{SABR_Calculations}. Finally, we take the limit $\epsilon \to 1$ to obtain the price approximation.
\end{proof}

To check the accuracy of \eqref{sabr_call_option_approximation}, we will use a Monte Carlo method and the implied volatility approximation proposed in \cite{Hagan02} for a range of strikes. The parameters used have been calibrated to the swaption market with the underlying swap tenor of 5 years and maturities of 5 years, 10 years, and 15 years.
\begin{table}[ht]
\centering
\begin{tabular}{||c c c c||}
 \hline
 SABR Parameters & 5Y & 10Y & 15Y \\ [0.5ex] 
 \hline\hline
 $\alpha$ & 0.0083 & 0.0075 & 0.0068 \\ 
 \hline
  $\nu$ & 0.335 & 0.243 & 0.215 \\
 \hline
  $\rho$ & 0.230 & 0.235 & 0.195 \\
 \hline
\end{tabular}
\caption{Normal SABR parameters considering an underlying swap tenor of 5 years.}
\label{table:SABRparams}
\end{table}


In figures \ref{fig:Call_5Y}, figures \ref{fig:Call_10Y}, and figures \ref{fig:Call_15Y} we compare the numerical performance of the Watanabe approximation against Hagan's approximation and a Monte Carlo simulation. We observe that the Watanabe approximation is slightly better than Hagan's. 

\begin{figure}[H]
  \includegraphics[width=1.0\linewidth]{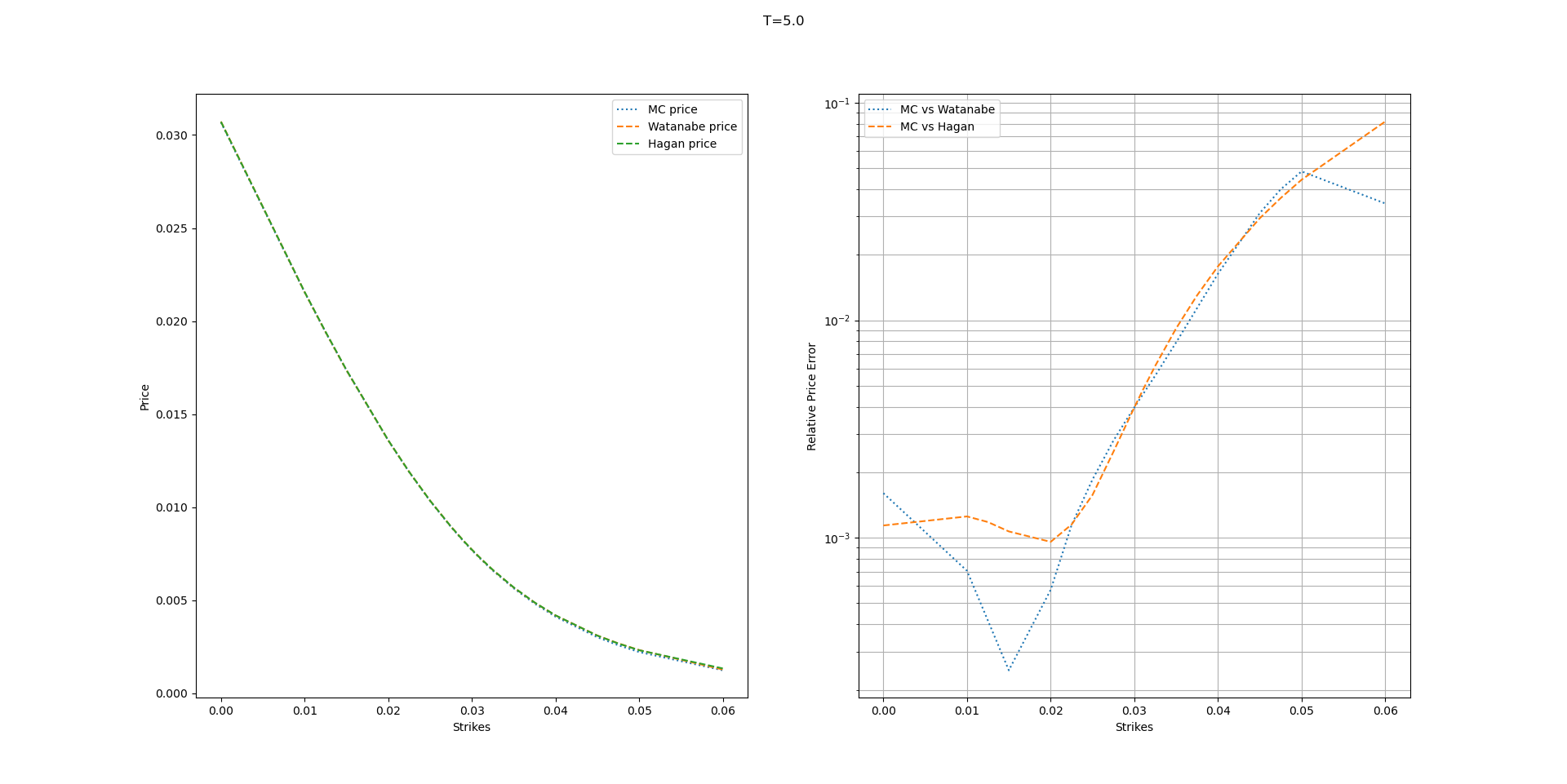}  
  \caption{Comparison of call options using Watanabe's expansion against the industry standards. The parameters considered are $\alpha=0.0083, \nu=0.335, \rho=0.23,$ and $T=5Y.$}
	\label{fig:Call_5Y}
\end{figure}  

\begin{figure}[H]
  \includegraphics[width=1.0\linewidth]{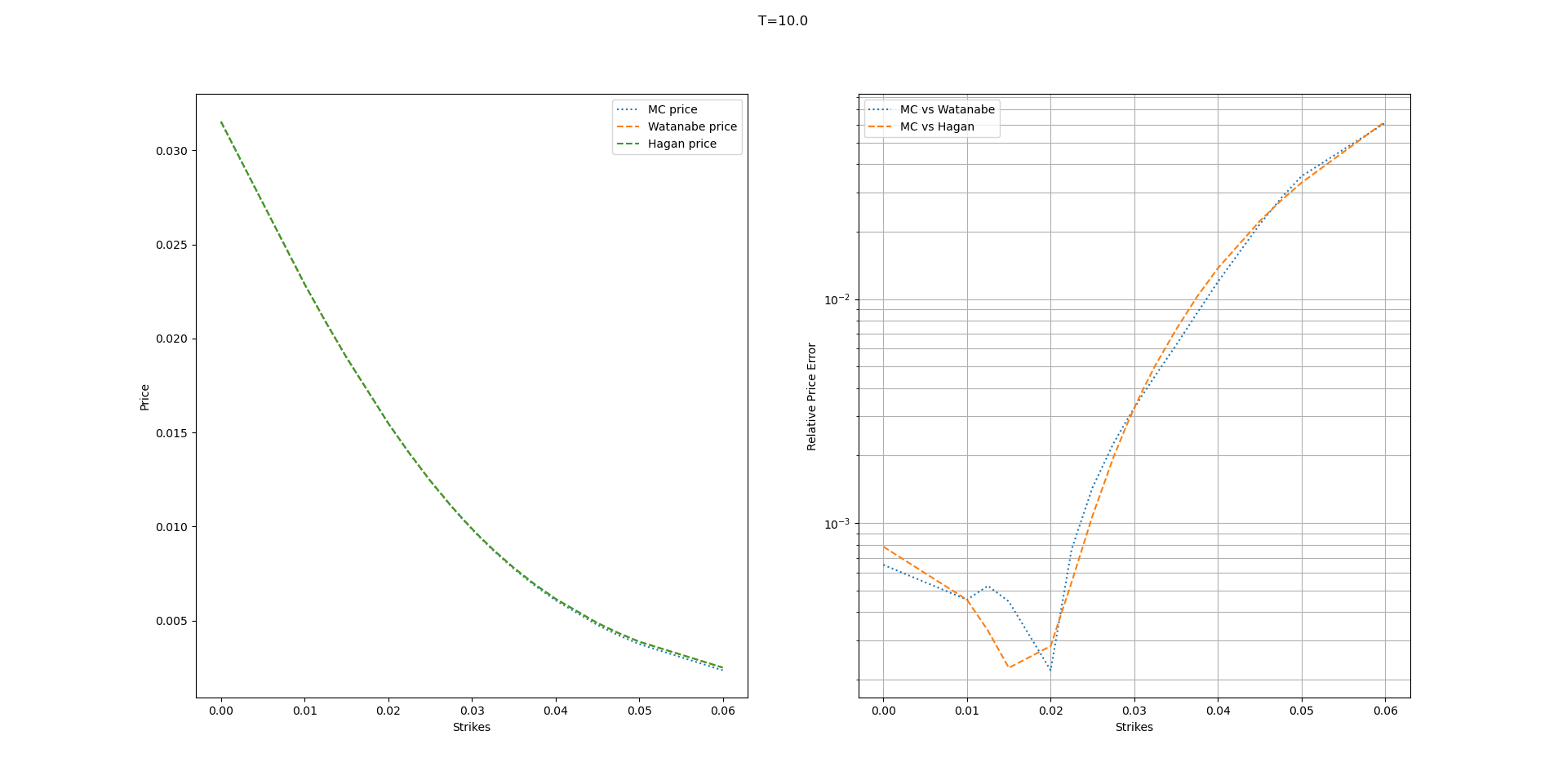}  
  \caption{Comparison of call options using Watanabe's expansion against the industry standards. The parameters considered are $\alpha=0.0075, \nu=0.243, \rho=0.235,$ and $T=10Y.$}
	\label{fig:Call_10Y}
\end{figure}

\begin{figure}[H]
  \includegraphics[width=1.0\linewidth]{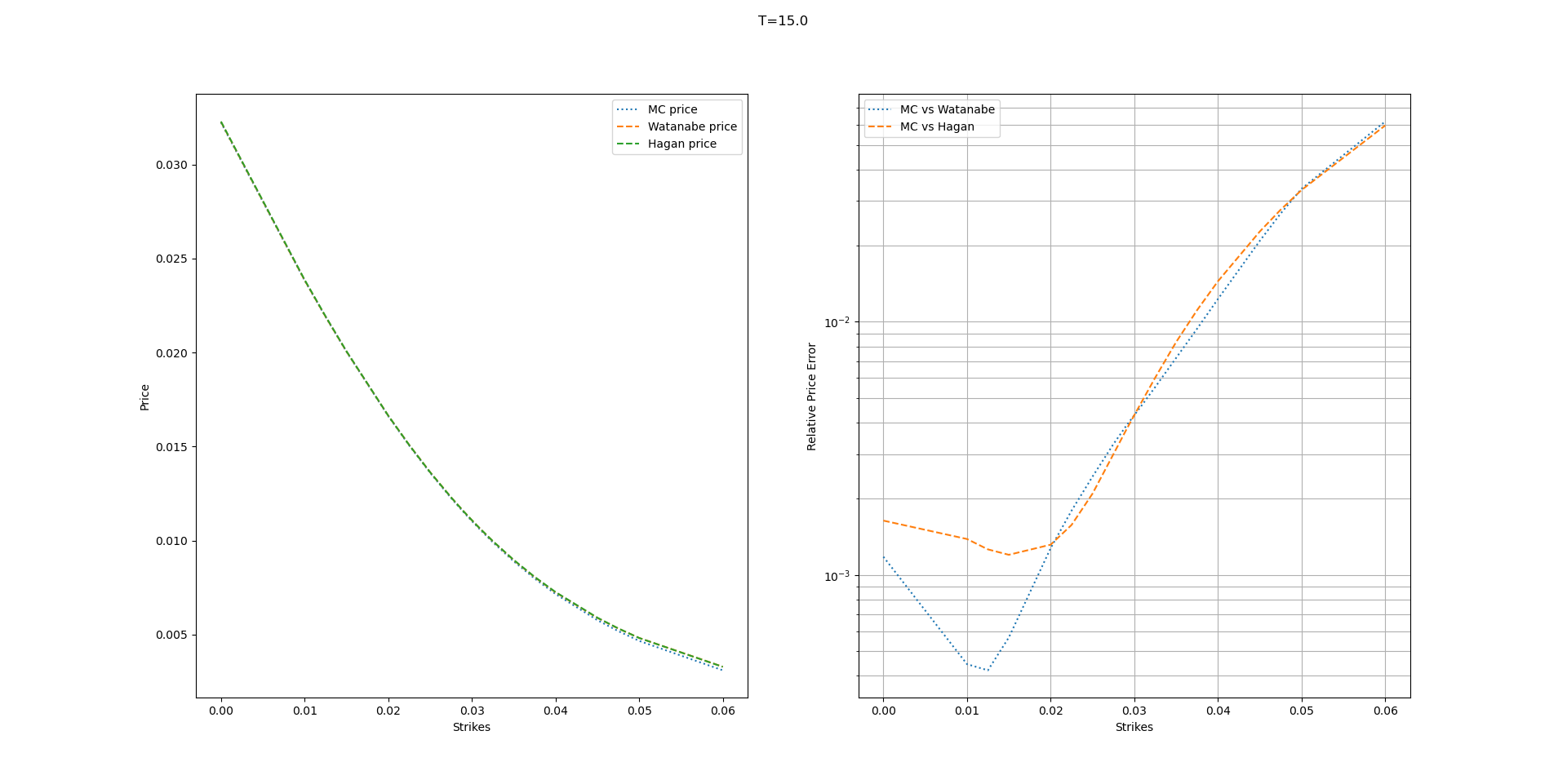}  
  \caption{Comparison of call options using Watanabe's expansion against the industry standards. The parameters considered are $\alpha=0.0068, \nu=0.215, \rho=0.195,$ and $T=15Y.$}
	\label{fig:Call_15Y}
\end{figure}

\end{example}


\section{Quadratic option price}\label{sec:QuadraticOptionPrice}
In this section, following the methodology presented in Section \ref{sec:Watanabe}, we will derive a Watanabe expansion for quadratic payoff derivatives when the underlying follows a stochastic volatility model, a local volatility model, and a local stochastic volatility. It is worth noting that the quadratic payoff $V^{QC}(F, K)$ is more influenced by the tail of the distribution than the call option price $V^{C}(T, F, K) $. In the case of local volatility, we will explain why this happens.\\

Consider that the forward price dynamics are given by 
\begin{equation}\label{forward_local_model}
dF_t = \sigma(F_t) dW^{N}_t
\end{equation}
where $\sigma(\cdot)$ is Lipschitz function to ensure existence and uniqueness of the solution. Now, if we use Itô-Tanaka formula to $\Big((F_t - K)^{+}\Big)^{2}$, we get that
\begin{equation}\label{qc_lv_representation}
\frac{V^{QC}(T,F,K)}{N_T} = \bigg((F_0 - K)^{+}\bigg)^{2} + \int_{0}^{T} \mathbb{E}^{\mathbb{Q}_N}\bigg[I\Big(F_u > K\Big) \sigma^{2}(F_u)\bigg] du.
\end{equation}
It is noteworthy that the value of $\frac{V^{QC}(T,F,K)}{N_T}$ depends on the value of $\sigma(F)$ when $F \in (K, \infty)$. For this reason, the quadratic call price is highly sensitive to our assumptions about the behavior of the local volatility function at extreme strikes.\\


As seen in Section \ref{sec:Watanabe}, we will find the Watanabe's price expansion for the quadratic call, put, and swap.

\begin{theorem}[Price Approximation for a Quadratic Call, Put, and Swaps]\label{QV expansions}
Let's suppose there exists a forward price expansion given by \eqref{base_expansion}. Let $y$ defined by \eqref{label_y}, and $I$ defined in \eqref{id_function}. Then, the $O(\epsilon^{3})$-approximation for Watanabe's expansion of a quadratic call option price is
\begin{align}\label{expansion_base_quadratic_call}
\frac{V^{QC}(T,F,K)}{N_T} &= \epsilon^{2} \sigma^{2}(F_0) T \mathbb{E}^{\mathbb{Q}^N}\Bigg[\bigg((\hat{g}_1(T) - y)^{+}\bigg)^{2} + 2 \epsilon \Big(\hat{g}_1(T) - y\Big)^{+} \hat{g}_2(T) \nonumber \\
&+ \epsilon^{2} \bigg( 2 \Big(\hat{g}_1(T) - y\Big)^{+} \hat{g}_3(T) + I\Big(\hat{g}_1(T) > y\Big)  \hat{g}^{2}_2(T) \bigg) + O(\epsilon^{3}) \Bigg].
\end{align}
The $O(\epsilon^{3})$-approximation of a quadratic put option price is
\begin{align}\label{expansion_base_quadratic_put}
\frac{V^{QP}(T,F,K)}{N_T} &= \epsilon^{2} \sigma^{2}(F_0) T \mathbb{E}^{\mathbb{Q}^N}\Bigg[ \bigg((y-\hat{g}_1(T))^{+}\bigg)^{2} - 2 \epsilon \Big(y - \hat{g}_1(T)\Big)^{+} \hat{g}_2(T) \nonumber \\
&+ \epsilon^{2} \bigg( I\Big(y > \hat{g}_1(T)\Big)  \hat{g}^{2}_2(T) - 2 \Big(y-\hat{g}_1(T)\Big)^{+} \hat{g}_3(T) \bigg) + O(\epsilon^{3})  \Bigg].
\end{align}
The $O(\epsilon^{3})$-approximation of a quadratic swap price is 
\begin{align}\label{expansion_base_quadratic_swap}
\frac{V^{QS}(T,F,K)}{N_T} &= \epsilon^{2} \sigma^{2}(F_0) T \mathbb{E}^{\mathbb{Q}^N}\Bigg[ \bigg((\hat{g}_1(T)- y) \bigg)^{2} + \epsilon^{2} \hat{g}^{2}_2(T) \nonumber  \\ 
& +2 \Big(\hat{g}_1(T) - y\Big) \hat{g}_3(T) + O(\epsilon^{3})   \Bigg].  
\end{align} 
\end{theorem}
\begin{proof}
We proceed as in \eqref{expansion_generic_payoff_order_5}. We use the forward expansion \eqref{base_expansion} to re-write the option price. Then, we use a Taylor's expansion on $V(T,F,K)$ around $\hat{g}_{1}(T) - y$ for each product.
\end{proof}

\subsection{Quadratic option under local volatility}
When the volatility is constant, that is $\sigma(\cdot)=\sigma$, the quadratic call option price is
\begin{equation}\label{bachelier_quadratic_option}
\frac{V^{QC}(T,F,K)}{N_T} = \bigg(\Big(F_0 - K\Big)^{2} + \sigma^2 T\bigg) \Phi\bigg( \frac{K - F_0}{\sigma \sqrt{T}}\bigg) + \Big(F_0 - K\Big) \sigma \sqrt{T} \phi\bigg(\frac{K - F_0}{\sigma \sqrt{T}}\bigg).
\end{equation}
where $\Phi(\cdot)$ and $\phi(\cdot)$ are, respectively, the density and cumulative function for a standard normal distribution. We are interested in building an option price expansion with base (\ref{bachelier_quadratic_option}). We will consider the small diffusion model
\begin{equation}\label{local_vol_model_small}
dF_{\epsilon,t} = \epsilon \sigma(F_{\epsilon,t}) dW^{N}_t
\end{equation}
where $\sigma(\cdot)$ is a positive function such that belongs to $\mathcal{C}^{3}_b(\mathbb{R})$.

\begin{lemma}\label{expansion_lv_underlying_lemma}
Under the dynamics given by \eqref{local_vol_model_small}, we can re-write the forward price dynamics as:
\begin{align}\label{expansion_lv_underlying}
F_{\epsilon,t} = F_0 &+ \epsilon \sigma(F_0)I_{(1)}(T) \nonumber \\
&+ \epsilon^2 \partial_{F}  \sigma(F_0) \sigma(F_0) I_{(1,1)}(T) \nonumber \\
&+ \epsilon^{3} \bigg(\Big( \partial^{2}_F \sigma(F_0) \sigma^{2}(F_0) + (\partial_F \sigma(F_0))^{2} \sigma(F_0)\Big)I_{(1,1,1)}(T) + \frac{1}{2} \partial^2_F \sigma(F_0) \sigma^{2}(F_0) I_{(0,1)}(T) \bigg) \nonumber \\
&+ E(T,\epsilon)
\end{align}
where under usual regularity conditions in $\sigma(\cdot)$ we have that $E(T,\epsilon) \sim O(\epsilon^{4})$.
\end{lemma}
\begin{proof}
The result is straightforward by iteratively applying the It\^o-Taylor expansion.
\end{proof}



As previously stated, to obtain the Watanabe's price expansion, we need to find an expression for the $\hat{g}_{i}$ terms. Comparing Lemma \ref{expansion_lv_underlying_lemma} with the price expansion \eqref{base_expansion}, we have that 

\begin{align}\label{coeff_lv_expansion}
\hat{g}_{1}(T) &= \frac{W_T}{\sqrt{T}} \nonumber \\
\hat{g}_{2}(T) &= \partial_F \sigma(F_0) \frac{I_{(1,1)(T)}}{\sqrt{T}} \nonumber \\
\hat{g}_{3}(T) &= \bigg( \frac{1}{2} \Big(\partial^{2}_F \sigma(F_0)\Big) \sigma(F_0) + \Big(\partial_F \sigma(F_0)\Big)^{2}\bigg) \frac{I_{(1,1,1)}(T)}{\sqrt{T}} + \frac{1}{2} \partial^2_F \sigma(F_0) \sigma(F_0) \frac{I_{(0,1)}(T)}{\sqrt{T}}.
\end{align}

\begin{theorem}[Local volatility expansion for quadratic payoffs]
Under the model \eqref{local_vol_model_small}, consider $y$ is given by \eqref{label_y}, $\phi(\cdot)$ is standard normal pdf, $\Phi(\cdot)$ the standard normal cdf and $\bar{\Phi}(\cdot)=1-\Phi(\cdot)$. Then, the 3rd order quadratic call option approximation price is
\begin{align}\label{LV_quadratic_option_call}
\frac{V^{QC}(T,F,K)}{N_T} = \sigma^{2}(F_0) T &\Bigg(G^{c}_{q}(y) + \partial_{F}\sigma(F_0) \phi(y) \sqrt{T} \nonumber \\
&+ \bigg( \frac{2}{3}\Big( \frac{\partial^{2}_F \sigma(F_0)}{2} \sigma(F_0) + \Big(\partial_F \sigma(F_0)\Big)^{2} \Big)y\phi(y) + \frac{1}{2}\partial^{2}_F\sigma(F_0)\sigma(F_0) \bar{\Phi}(y)    \bigg) T   \nonumber \\
&+ \frac{1}{4}\Big(\partial_F \sigma(F_0)\Big)^{2}  \bigg(\Big(y^3+y\Big)\phi(y)+ 2\bar{\Phi}(y)\bigg)T + O(\epsilon^{3}) \Bigg)
\end{align}
where $G^{c}_{q}(y) = (1+y^2) \bar{\Phi}(y) - y\phi(y)$.\\
The 3rd order quadratic put option approximation price is
\begin{align}\label{LV_quadratic_option_put}
\frac{V^{QP}(T,F,K)}{N_T} = \sigma^{2}(F_0) T &\Bigg(G^{p}_{q}(y) - \partial_{F}\sigma(F_0) \phi(y) \sqrt{T} \nonumber \\
&- \bigg( \frac{2}{3}\Big(  \frac{\partial^{2}_F \sigma(F_0)}{2} \sigma(F_0) + \Big(\partial_F \sigma(F_0)\Big)^{2} \Big)y\phi(y) - \frac{1}{2}\partial^{2}_F\sigma(F_0)\sigma(F_0) \Phi(y) \bigg) T   \nonumber \\
&+ \frac{1}{4}\Big(\partial_F \sigma(F_0)\Big)^{2}  \bigg(2\Phi(y) - \Big(y^3+y\Big)\phi(y)\bigg)T + O(\epsilon^{3}) \Bigg)
\end{align}
where $G^{p}_{q}(y) = (1+y^2) \Phi(y) + y\phi(y)$.\\
Finally, the quadratic swap price is 
\begin{align}\label{LV_quadratic_option_swap}
\frac{V^{QS}(T,F,K)}{N_T} = \sigma^{2}(F_0) T \Bigg(\Big(1+ y^2\Big) + \frac{1}{2}\bigg( \partial^{2}_F\sigma(F_0)\sigma(F_0)+ \Big(\partial_F\sigma(F_0)\Big)^2 \bigg) + O(\epsilon^3) \Bigg).
\end{align}
\end{theorem}
\begin{proof}
We consider the quadratic price expansion of Theorem \ref{QV expansions}. In the quadratic call option price, unifying \eqref{epsilon_0_quadratic_call}, \eqref{epsilon_1_quadratic_call} and \eqref{epsilon_2_quadratic_call} along with  \eqref{expansion_lv_underlying}. Then, we take the limit $\epsilon \to 1$ to obtain the quadratic call option price approximation. We proceed similarly with the quadratic put option price approximation and the quadratic swap option price. 
\end{proof}

\begin{remark}
It is important to note that the last expressions satisfy parity for quadratic options, that is
\begin{equation*}
V^{QS}(T,F,K) = V^{QP}(T,F,K) + V^{QC}(T,F,K).
\end{equation*}
\end{remark}

\begin{example}\label{sabr_lv_example}
To verify the accuracy of  \eqref{LV_quadratic_option_call}, we will apply it to the normal SABR model, i.e. $\beta=0$. In this case, the equivalent local volatility has a closed-form solution, given by \cite{Balland}. The expression for the equivalent local volatility is as follows:

\begin{equation} \label{local_vol_normal_sabr}
\sigma(K) = \alpha \sqrt{1+2\rho \nu \left(\frac{K - F_0}{\alpha}\right) + \nu^{2} \left(\frac{K - F_0}{\alpha}\right)^{2}}.
\end{equation}
We will compare the accuracy of \eqref{LV_quadratic_option_call} with the Monte Carlo method and the approximation proposed in  (\cite{Hagan19}) for a range of strikes and maturities. As we have done in Section \ref{sec:Watanabe}, we use SABR parameters calibrated to the swaption market.
The parameters used are listed in Table \ref{table:SABRparams}.

In this case, we can see that
\begin{align*}
\partial_{F}\sigma(F_0) &= \rho \nu ,\nonumber \\
\partial^{2}_{F}\sigma(F_0) &= \frac{\nu^{2}}{\alpha} (1-\rho^{2})= \frac{\nu^{2} \hat{\rho}^{2}}{\alpha}.
\end{align*}

When we replace the above equalities in the Watanabe's call option expansion \eqref{LV_quadratic_option_call}, we obtain the following approximation for $V^{QC}(K)$.
\begin{align}\label{quadratic_call_price_sabr_normal}
\frac{V^{QC}(F,K)}{N_T} = \alpha^2 T \Bigg(G_{q}(y)  &+ \rho \nu \phi(y) \sqrt{T} +
\frac{1}{6}\bigg( 2y \phi(y) + 3\hat{\rho}^2  \bar{\Phi}(y) \bigg)\nu^{2}T  \nonumber \\
&+  \frac{1}{4} \rho^2 \nu^2  \bigg(\Big(y^{3}+y\Big) \phi(y) + 2 \bar{\Phi}(y)\bigg) T + O(\epsilon^3) \Bigg).
\end{align}

In this analysis, we compare the numerical performance of the Watanabe approximation against Hagan's approximation and a Monte Carlo simulation. The comparison is presented in figures \ref{fig:Quadratic_5Y}, \ref{fig:Quadratic_10Y}, and \ref{fig:Quadratic_15Y}. Our observation is that Hagan's approximation yields better results than Watanabe's approximation. We also noticed that the accuracy of Watanabe's approximation does not vary with maturity and declines linearly within the strike.

\begin{figure}[H]
  \includegraphics[width=1.0\linewidth]{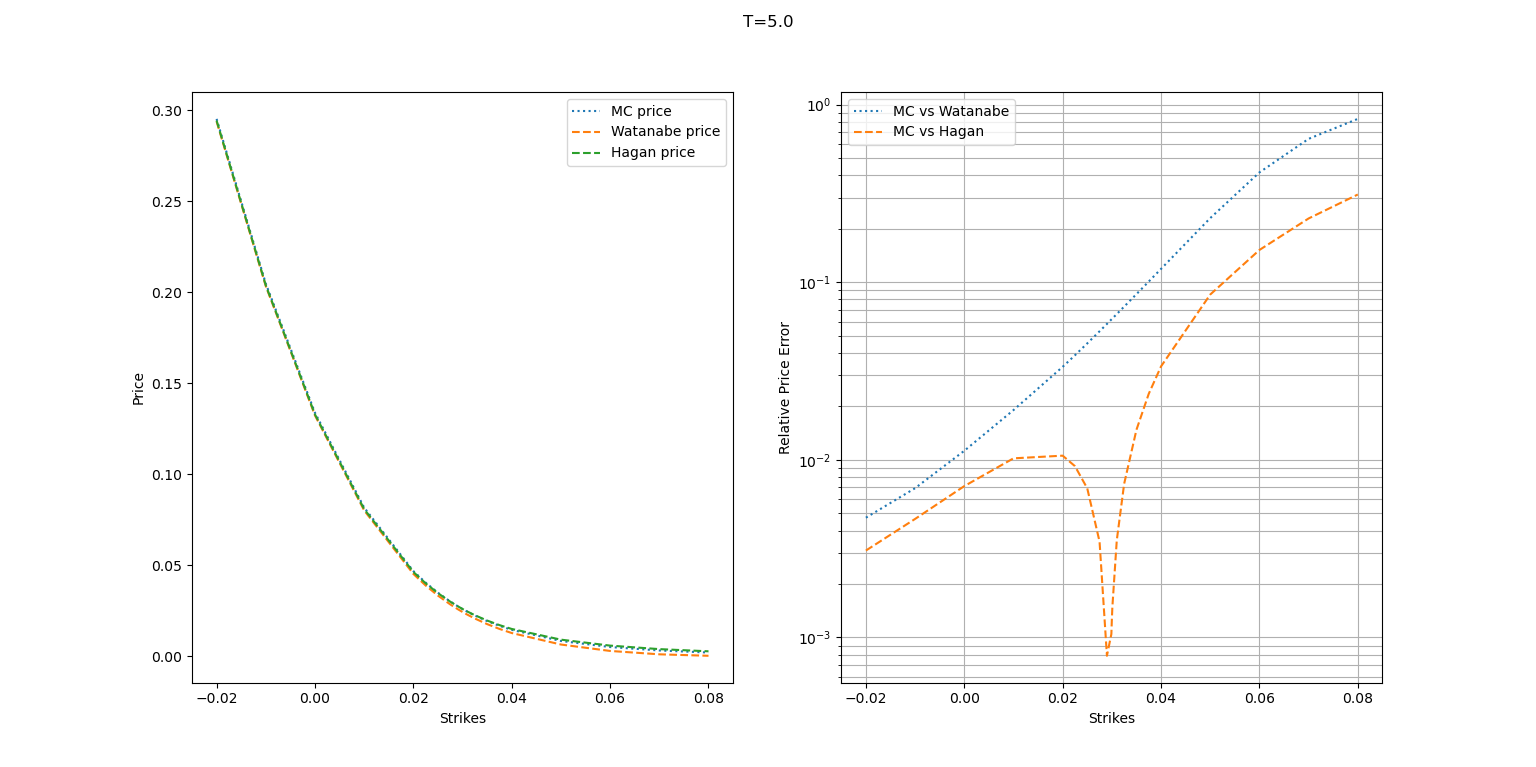}  
  \caption{Comparison of quadratic call options using Watanabe's expansion against the industry standards. The parameters considered are $\alpha=0.0083, \nu=0.335, \rho=0.23,$ and $T=5Y.$}
	\label{fig:Quadratic_5Y}
\end{figure}  

\begin{figure}[H]
  \includegraphics[width=1.0\linewidth]{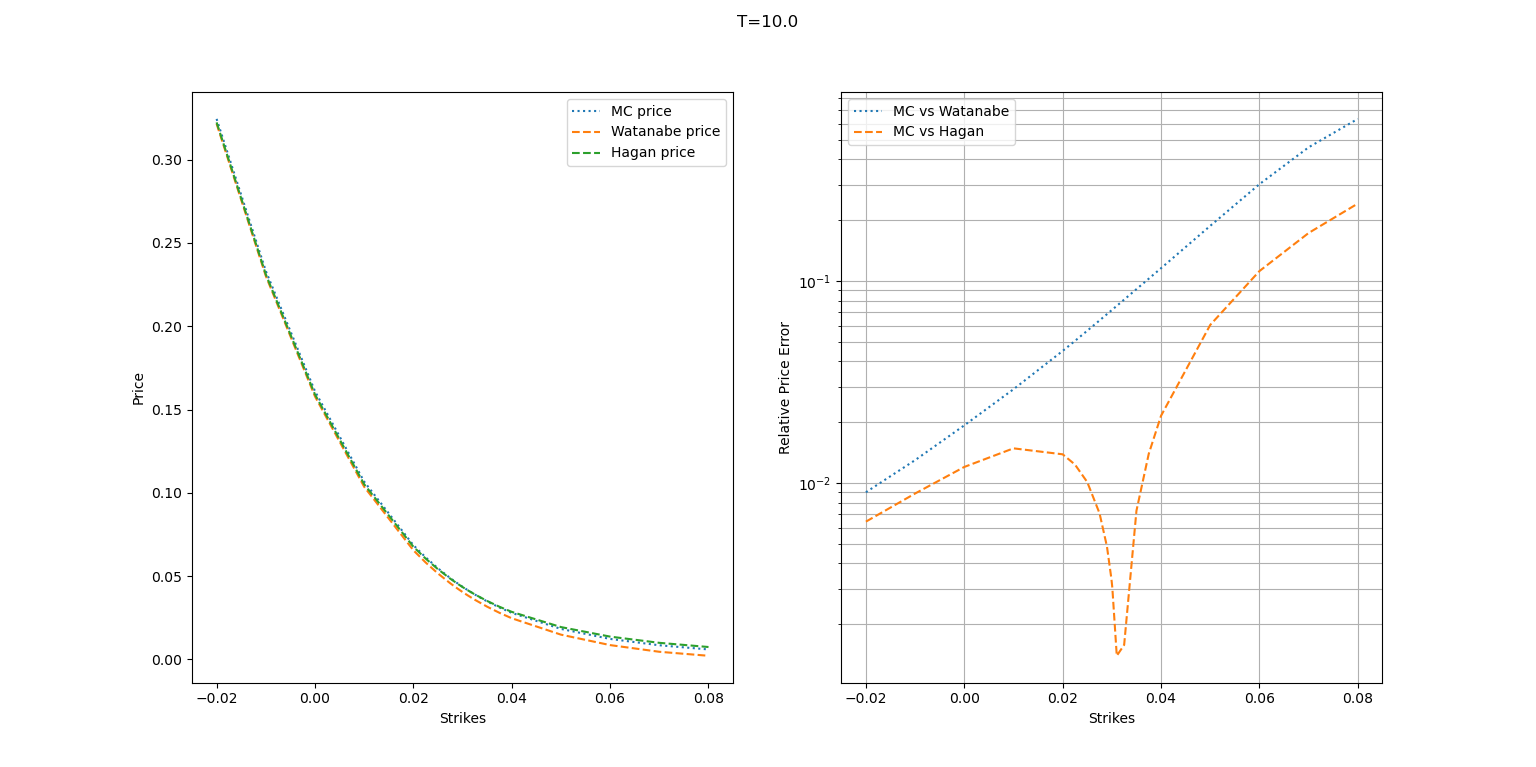}  
  \caption{Comparison of quadratic  call options using Watanabe's expansion against the industry standards. The parameters considered are $\alpha=0.0075, \nu=0.243, \rho=0.235,$ and $T=10Y.$}
	\label{fig:Quadratic_10Y}
\end{figure}

\begin{figure}[H]
  \includegraphics[width=1.0\linewidth]{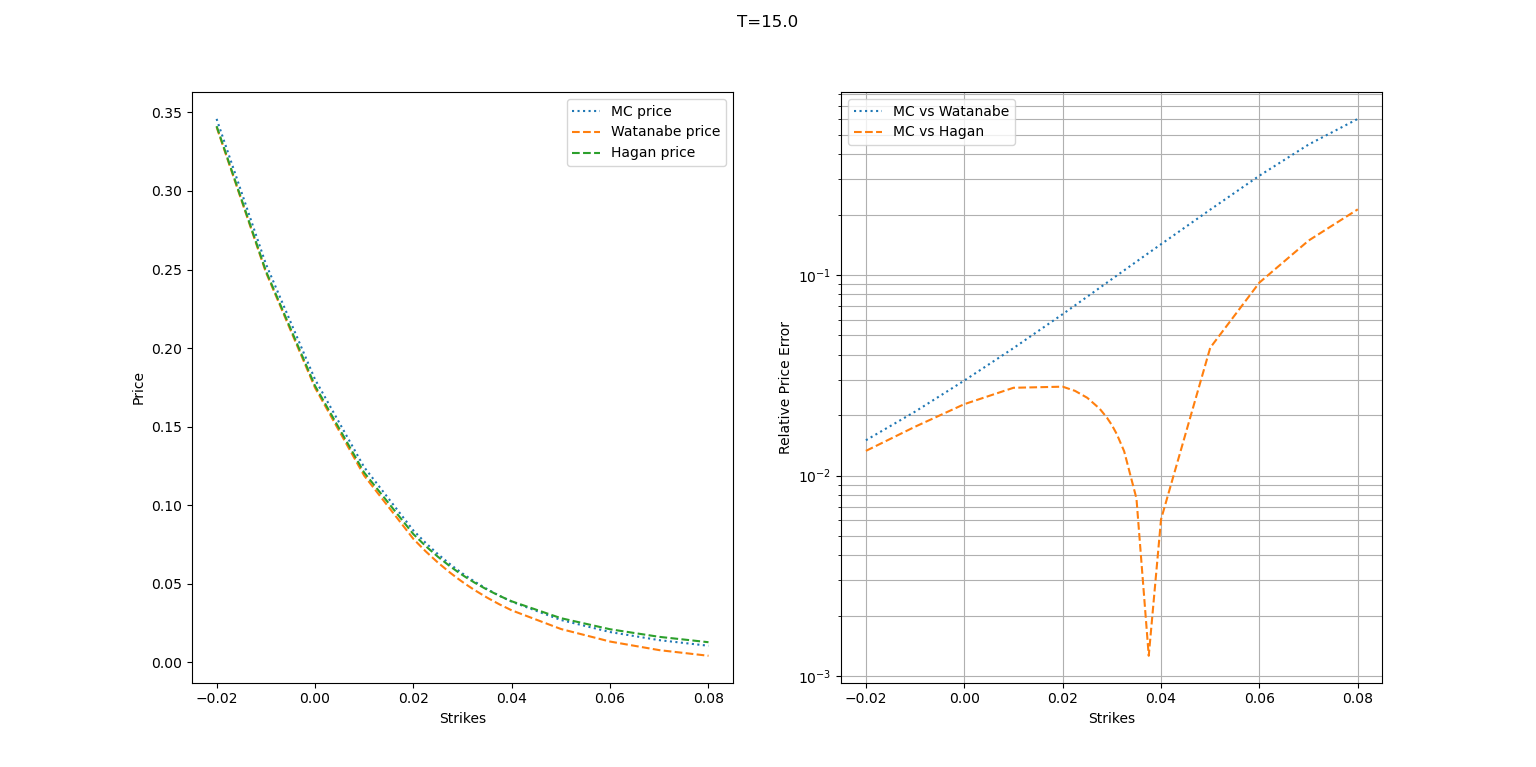}  
  \caption{Comparison of quadratic call options using Watanabe's expansion against the industry standards. The parameters considered are $\alpha=0.0068, \nu=0.215, \rho=0.195,$ and $T=15Y.$}
	\label{fig:Quadratic_15Y}
\end{figure}


\end{example}

\subsection{Quadratic option under stochastic local volatility}
SLV-type models are widely used in the industry; see, for example, \cite{Bang2019} and \cite{Labordere}. These models allow for control of the backbone through a layer of local volatility and its movement over time. In the interest rate market, a SABR dynamic is the market standard. We will consider a general SABR dynamics given by:
\begin{align}\label{general_SLV_mode}
dF_t &= C(F_t) \sigma_t dW_t, \nonumber \\
d\sigma_t &= \nu \sigma_t dB_t.
\end{align}
Here, $<dW_t,dB_t>= \rho dt$ and $C(\cdot)$ is a positive function that belongs to $\mathcal{C}^{3}_b(\mathbb{R})$. To achieve an expansion for $F_t$ similar to \eqref{expansion_lv_underlying}, we will consider the next parametrization of the above dynamic
\begin{align}\label{general_SLV_mode_small_e}
dF_{\epsilon, t} &= \epsilon  C(F_{\epsilon, t}) \sigma_{\epsilon, t} dW_t, \nonumber \\
d\sigma_{\epsilon, t} &= \epsilon  \nu \sigma_{\epsilon, t} dB_t .
\end{align}

\begin{lemma}\label{expansion_slv_underlying_lemma}
Under the dynamics given by \eqref{general_SLV_mode_small_e}, we can re-write the forward price dynamics as:
\begin{align}\label{expansion_slv_underlying}
F_{\epsilon,t} = F_0 &+ \epsilon \sigma(F_0)I_{(1)}(T) \nonumber \\
&+ \epsilon^2 \left(\partial_F C(F_0) \alpha I_{(1,1)}(T)+ \nu \int_{0}^{T}\int_{0}^{u_1}  dB_{u_2} dW_{u_2}\right) \nonumber \\
&+ \epsilon^{3} \bigg( \nu^2 \int_0^T \int_0^{u_1}\int_0^{u_2} dB_s dB_{u_2} dW_{u_1} + \partial_F C(F_0) \alpha \nu  \int_{0}^{T} B_u W_u dW_u \nonumber \\
&+  \frac{1}{2}\alpha^2 \partial^2_F C(F_0)C(F_0) \int_{0}^{T} W^2_u dW_u + \alpha  \nu \partial_F C(F_0) \int_0^{T} \int_0^{u_1} \int_0^{u_2} dB_s W_{u_2} dW_{u_1}\nonumber \\
&+ \alpha^{2} \partial_{F}C(F_0) I_{(1,1,1)}(T) \bigg) + E(T,\epsilon)
\end{align}
where under usual regularity conditions in $\sigma(\cdot)$ we have that $E(T,\epsilon) \sim O(\epsilon^{4})$.
\end{lemma}
\begin{proof}
The result is straightforward by iteratively applying the It\^o-Taylor expansion.
\end{proof}

Then, as we have seen before, we can imply the $\hat{g}_{i}$ terms.
\begin{align}\label{g_s_slv}
\hat{g}_1(T) &= \frac{W_T}{\sqrt{T}} \nonumber \\
\hat{g}_2(T) &= \partial_F C(F_0)  \frac{\alpha}{\sqrt{T}}I_{(1,1)}(T) +  \frac{\nu }{\sqrt{T}}\int_{0}^{T}\int_{0}^{u_1}  dB_{u_2} dW_{u_2}  \nonumber \\
\hat{g}_3(T) &=  \frac{\nu^2}{\sqrt{T}}\int_0^T \int_0^{u_1}\int_0^{u_2} dB_s dB_{u_2} dW_{u_1} + \frac{\alpha \nu}{\sqrt{T}} \partial_F C(F_0)   \int_{0}^{T} B_u W_u dW_u \nonumber \\
&+  \frac{\alpha^2 }{2\sqrt{T}}\Big(\partial^2_F C(F_0)\Big)C(F_0) \int_{0}^{T} W^2_u dW_u +  \frac{\alpha  \nu}{\sqrt{T}}\partial_F C(F_0) \int_0^{T} \int_0^{u_1} \int_0^{u_2} dB_s W_{u_2} dW_{u_1} \nonumber \\
&+ \alpha^{2} \partial_{F}C(F_0) \frac{I_{(1,1,1)}(T)}{\sqrt{T}}.
\end{align}

\begin{theorem}[Stochastic local volatility expansion for quadratic payoffs]
Under the model \eqref{general_SLV_mode_small_e}, consider $y$ is given by \eqref{label_y}, $\phi(\cdot)$ is standard normal pdf, $\Phi(\cdot)$ the standard normal cdf and $\bar{\Phi}(\cdot)=1-\Phi(\cdot)$. Then, the 3rd order quadratic call option approximation price is
\begin{align}\label{SLV_quadratic_option_call}
\frac{V^{QC}(T,F,K)}{N_T}   = \alpha^2C^2(F_0)  T &\Bigg(G_{q}(T) +  \bigg(\nu \rho + \frac{1}{2}\alpha \partial_F  C(F_0) \bigg)  \phi(y) \sqrt{T}      \nonumber \\
&+\frac{2}{3}\bigg(\Big(\frac{1}{2}\nu \rho + \alpha \partial_F C(F_0)  \Big)\rho  \nu + \alpha^{2} \Big(\partial_{F}C(F_0)\Big)^2 \bigg) y\phi(y)T   \nonumber \\
&+ \alpha\bigg( \nu \partial_F C(F_0) +  \frac{1}{2} \alpha \partial^2_F C(F_0)C(F_0) \bigg) \bigg(\frac{2}{3}y\phi(y) + \bar{\Phi}(y) \bigg) T   \nonumber \\
&+ \bigg(\frac{1}{4}\Big(\alpha^2 \big(\partial_F C(F_0)\big)^2+ \nu^{2} \rho^{2} \Big) \Big(2\bar{\Phi}(y)+(y^3+y)\phi(y) \Big) \nonumber \\
& \hspace{0.75cm}+\frac{1}{6}\hat{\rho}^2 \nu^2\Big(2 y \phi(y) + 3 \bar{\Phi}(y)\Big) \bigg) T + O(\epsilon^3)\hspace{0.2cm}\Bigg)
\end{align}
where $G^{c}_{q}(y) = (1+y^2) \bar{\Phi}(y) - y\phi(y)$.\\
The 3rd order quadratic put option approximation price is
\begin{align}\label{SLV_quadratic_option_put}
\frac{V^{QP}(T,F,K)}{N_T}   = \alpha^2 C^2(F_0)  T &\Bigg(G_{q}(T) -  \bigg(\nu \rho + \frac{1}{2}\alpha \partial_F  C(F_0) \bigg)  \phi(y) \sqrt{T}      \nonumber \\
&- \frac{2}{3}\bigg(\Big(\frac{1}{2}\nu \rho + \alpha  \partial_F C(F_0)\Big)\rho  \nu  + \alpha^{2} \Big(\partial_{F}C(F_0)\Big)^2 \bigg) y\phi(y)T  \nonumber \\
&- \alpha\bigg( \nu \partial_F C(F_0) +  \frac{1}{2}\alpha \partial^2_F C(F_0)C(F_0)\bigg) \bigg(\frac{2}{3}y\phi(y) - \Phi(y) \bigg) T   \nonumber \\
&+ \bigg(\frac{1}{4}\Big(\alpha^2 (\partial_F C(F_0))^2+ \nu^{2} \rho^{2} \Big) \Big(2\Phi(y)- (y^3+y)\phi(y) \Big)  \nonumber \\
& \hspace{0.75cm}- \frac{1}{6}\hat{\rho}^2 \nu^2\Big(2 y \phi(y)+ 3\Phi(y)\Big) \bigg) T + O(\epsilon^3)\hspace{0.2cm}\Bigg)
\end{align}
where $G^{p}_{q}(y) = (1+y^2) \Phi(y) + y\phi(y)$.\\
Finally, the 3rd order quadratic swap price is 
\begin{align}\label{SLV_quadratic_swap}
\frac{V^{QS}(T,F,K)}{N_T}  = \alpha^2 C^2(F_0)  T  \Bigg((1+y^2) + \bigg(&\alpha \nu \partial_F C(F_0)  + \frac{1}{2}\alpha^{2}\Big( \partial^2_F C(F_0)C(F_0)+ \Big(\partial_F C(F_0)\Big)^2  \Big)\nonumber \\
 &  + \frac{1}{2} \nu^{2}\bigg) T + O(\epsilon^3) \Bigg).
\end{align}
\end{theorem}
\begin{proof}
We consider the quadratic price expansion of Theorem \ref{QV expansions} and the conditional expectations values calculated in \eqref{appendix:b11}, \eqref{appendix:b12} and \eqref{appendix:b13}. Then, the quadratic call option price, unifying \eqref{appendix:e51}
\eqref{appendix:e52}, and \eqref{appendix:e53} along with  Lemma \ref{expansion_slv_underlying_lemma}. Then, we take the limit $\epsilon \to 1$ to obtain the quadratic call option price approximation. We proceed similarly with the quadratic put option price approximation and the quadratic swap option price. 
\end{proof}

\begin{remark}
We must note that when $\nu=0, \rho=0$ and $\alpha=1$, the model \eqref{general_SLV_mode} is a pure local volatility model. But, when we substitute these parameters in \eqref{SLV_quadratic_option_call}, then we recover exactly \eqref{LV_quadratic_option_call}.
\end{remark}

\begin{example}
If $C(F)=1$, then we recover the normal SABR model. For this case, $\partial_F C(F)=0$ and $\partial^2_F C(F)=0$. Therefore \eqref{SLV_quadratic_option_call} is reduced to \eqref{quadratic_call_price_sabr_normal}), which is consistent with the fact that the local volatility is the equivalent local volatility for the normal SABR model.
\end{example}

\section{Future works}\label{sec:future_works}
In a recent paper \cite{Skoufis2024}, the author proposes a closed formula for the convexity adjustment in the case of average RFR indexes when the underlying follow a time-dependent SABR model. This problem was addressed in  \cite{DavidRaul2023} in the case of the Cheyette model with local volatility using Malliavin calculus. To make the paper self-contained, let us introduce the basic concepts. Let us define
\begin{align*}
I(T_a,T_b)&=\mathbb{E}^{\mathbb{Q}}\left[\frac{e^{\int_{T_a}^{T_b} r_u du} - 1}{T_b -T_a}  \right],\\
I_{avg}(T_a,T_b)&=\frac{\int_{T_a}^{T_b} r_u du}{T_b - T_a},
\end{align*}\label{sabr_time_dependent_rfr}
and $R(t)=\mathbb{E}^{\mathbb{Q}}\left[I(T_a,T_b) \right]$. In addition, we will suppose that
\begin{align}
dR_{t}(T_a,T_b) &= h_t C(R_{t}(T_a,T_b)) \sigma_{t} dW_t, \nonumber \\
d\sigma_{t} &= \nu \sigma_{t} dB_t \nonumber \\
<dW_t,dB_t> &= \rho dt
\end{align}
where $h_t = \min \left\{1, \left(\frac{T_b - t}{T_b - T_a} \right)^{q} \right\}$ with $0 < q < 1$. From the definition of $I_{avg}(T_a,T_b)$, we have that
\begin{equation*}
I_{avg}(T_a,T_a,T_b) = \frac{\ln\Big(1 + (T_b - T_a) I(T_a,T_b)\Big)}{T_b - T_a} = \frac{\ln\Big(1 + (T_b - T_a) R(T_a, T_a,T_b)\Big)}{T_b - T_a}.
\end{equation*}
Then, the issue is to compute $\mathbb{E}\left[ I_{avg}(T_a,T_b) \right]$ when $R_{t,T_a,T_b}$ follow the dynamic \eqref{sabr_time_dependent_rfr}. Using Itô formula, we obtain that
\begin{align*}
\mathbb{E}\left[ I_{avg}(T_a,T_b) \right] &= I(0,T_a,T_b) + \frac{1}{2} \mathbb{E}^{\mathbb{Q}} \left[\int_{0}^{T_b} G^{\prime \prime}(R(u,T_a,T_b)) h_u^{2} \sigma^{2}_u C^{2}(R(u,T_a,T_b)) ds \right] \\
& \approx I(0,T_a,T_b) + \frac{1}{2} \mathbb{E}^{\mathbb{Q}} G^{\prime \prime}(R(0,T_a,T_b)) \mathbb{E}^{\mathbb{Q}}\bigg[ \Big(R(T_b,T_a,T_b) - R(0,T_a,T_b)\Big)^2 \bigg]
\end{align*}
where 
\begin{align*}
G(R)=\frac{\ln(1 + (T_b - T_a) R)}{T_b - T_a}.
\end{align*}
 Let us suppose that we have an expansion similar to \eqref{base_expansion}. Then, it is easy to show that
\begin{align*}
\mathbb{E}^{\mathbb{Q}}&\bigg[ \Big(R(T_b,T_a,T_b) - R(0,T_a,T_b)\Big)^2 \bigg]  \\
&= \alpha^2 C^2(R(0,T_a,T_b)) \mathbb{E}^{\mathbb{Q}}\bigg[ \hat{g}^2_{1}(T_b) +  \hat{g}^2_{1}(T_b) + 2\hat{g}_{1}(T_b) \hat{g}_{3}(T_b) + O(\epsilon^3)\bigg]  T.  
\end{align*}
The challenge for future papers is to obtain the expansion of $F_{\epsilon, t}$ under a general SLV dynamic, this way could compute the convexity adjustment of CMS or average RFR, even the pricing of generic terminal payoff for a general SLV model.

\section{Conclusions}\label{sec:Conclusion}
In conclusion, we have proposed a new approach for pricing CMS derivatives. Our method uses Mallaivin's calculus and Watanabe's expansions. Firstly, we establish a model-independent connection between the price of a CMS derivative and the quadratic payoffs using Mallaivin calculus. Then, we extend the results of Watanabe's expansions to the quadratic payoff case under local and stochastic local volatility. The approximations obtained are generic. To evaluate its accuracy, we compare the approximations numerically under the SABR model against the industry standards: Hagan's approximation, and a Monte Carlo simulation. Our approximations are as precise as Hagan's approximation, but they cover a wider range of models.

\section*{Appendix}
\appendix
We will consider $W_{t}$, $W_{\perp,t}$ two independent Brownian motions.
\renewcommand{\thesection}{\Alph{section}.\arabic{section}}

\section{Normal SABR}\label{SABR_Calculations}
\subsection{Calculation of $\mathbb{E}^{\mathbb{Q}^{N}}\left[ \hat{g}_{2}(T) |\hat{g}_1(T)\right]$}
\label{appendix:a11}
We have that 
\begin{equation*}
\mathbb{E}^{\mathbb{Q}^{N}}\Bigg[ \hat{g}_{2}(T) \Big|\hat{g}_1(T)\Bigg]=\mathbb{E}^{\mathbb{Q}^{N}}\Bigg[\frac{\nu}{\sqrt{T}} \bigg(\rho \int_{0}^{T} W_{u} dW_{u} + \bar{\rho} \int_{0}^{T} W_{\perp,u} dW_{u}\bigg) \Big| \hat{g}_{1}(T)\Bigg].
\end{equation*}
Therefore,
\begin{equation}
\mathbb{E}^{\mathbb{Q}^{N}}\Bigg[ \hat{g}_{2}(T) \Big|\hat{g}_1(T)\Bigg]= \frac{\rho \nu}{\sqrt{T}} \mathbb{E}^{\mathbb{Q}^{N}}\Bigg[  \int_{0}^{T} W_{u} dW_{u} \Big| \hat{g}_{1}(T) \Bigg] = \frac{1}{2}\rho \nu\sqrt{T}\Bigg(\hat{g}^{2}_{1}(T) - 1\Bigg).
\end{equation}

\subsection{Calculation of $\mathbb{E}^{\mathbb{Q}^{N}}\left[ \hat{g}_{3}(T) |\hat{g}_1(T)\right]$}
\label{appendix:a12}
We will follow the same approach as in the previous section. Therefore, we have
\begin{eqnarray*}
\hat{g}_{3}(T) &=& \frac{\nu^{2}}{2\sqrt{T}} \int_{0}^{T} B_{u}^{2} - u dW_{u} \\
&=& \frac{\nu^{2}}{2\sqrt{T}} \int_{0}^{T} B_{u}^{2} dW_{u} - \frac{\nu^{2}}{2\sqrt{T}} \int_{0}^{T} u dW_{u}.
\end{eqnarray*}
In one hand, we obtain that
\begin{align}
\mathbb{E}^{\mathbb{Q}^{N}}\Bigg[B(T)\Big|W_{T} = \hat{g}_{1}(T) \sqrt{T} \Bigg] =
\frac{1}{4}\nu^{2} \hat{g}_{1}(T) T.
\end{align}
On other hand, we have that
\begin{align}
\frac{\rho^{2} \nu^{2}}{2\sqrt{T}} \mathbb{E}^{\mathbb{Q}^{N}}\Bigg[\int_{0}^{T} W_{u} dW_{u} \Big|W_{T} = \hat{g}_{1}(T) \sqrt{T} \Bigg] 
&= \frac{1}{12}\rho^{2} \nu^{2} \Bigg( 2 \hat{g}^{3}_{1}(T) - \hat{g}_{1}(T) \Bigg) T. \\
\frac{\hat{\rho}^{2} \nu^{2}}{2\sqrt{T}} \mathbb{E}^{\mathbb{Q}^{N}}\Bigg[\int_{0}^{T} W^{2}_{u,\perp} dW_{u} \Big|W_{T} = \hat{g}_{1}(T) \sqrt{T} \Bigg] &= 
\frac{1}{4}\hat{\rho}^{2} \nu^{2}\hat{g}_{1}(T)T.
\end{align}
Therefore, when combine it, we get that
\begin{equation}\label{g3_g1}
\mathbb{E}^{\mathbb{Q}^{N}}\Bigg[\hat{g}_{3}(T) \Big| \hat{g}_{1}(T) \Bigg] = \frac{1}{12} \rho^{2} \nu^{2} \Bigg( 2 \hat{g}^{3}_{1}(T) - \hat{g}_{1}(T) \Bigg) T + \frac{1}{4}\hat{\rho}^{2} \nu^{2}\hat{g}_{1}(T) T  - \frac{1}{4}\nu^{2} \hat{g}_{1}(T) T.
\end{equation}

\subsection{Calculation of $\mathbb{E}^{\mathbb{Q}^{N}}\left( \hat{g}^{2}_{2}(T) |\hat{g}_1(T)\right)$}
\label{appendix:a13}
From the definition of $\hat{g}_{2}(T)$, we have that
\begin{align*}
\frac{\nu^{2}}{T} \Bigg(\underbrace{\rho^{2} \bigg(\int_{0}^{T} W_{u} dW_u\bigg)^{2}}_{A(T)} + \underbrace{\bar{\rho}^{2} \bigg(\int_{0}^{T} W_{\perp, u} dW_u\bigg)^{2}}_{B(T)} + \underbrace{2 \rho \bar{\rho} \bigg(\int_{0}^{T} W_{\perp, u} dW_u\bigg) \bigg(\int_{0}^{T} W_{u} dW_u\bigg)}_{C(T)} \Bigg)
\end{align*}
Now, we will compute $\mathbb{E}^{\mathbb{Q}^{N}}\Big[\cdot | \hat{g}_1(T)\Big]$ for each term. The easiest one is $\mathbb{E}^{\mathbb{Q}^{N}}\Big[C(T)|\hat{g}_{1}(T)\Big]$, due to the independence of $W_{T}$ and $W_{\perp,t}$. We get 
$$
\mathbb{E}^{\mathbb{Q}^{N}}\Big[C(T)|\hat{g}_{1}(T)\Big] =  0.
$$
Then, we will compute the rest of the terms
\begin{equation}\label{g_2_2_first_term}
\frac{\nu^{2} \rho^{2}}{T} \mathbb{E}^{\mathbb{Q}^{N}} \Bigg[\bigg(\int_{0}^{T} W_{u} dW_{u}\bigg)^{2} \Big|W_{T} = \hat{g}_{1}(T) \sqrt{T}\Bigg] 
= \frac{1 }{4}\nu^{2} \rho^{2} T\Bigg(\hat{g}^{2}_{1}(T) - 1 \Bigg)^{2}.
\end{equation}

For the last term, we must remember the next identity for a Brownian bridge $X_s$ with $X_{0} = 0$ and $X_{T} = \hat{g}_{1}(T)\sqrt{T}$
\begin{equation}\label{brownian_bridge_dw_dw}
\mathbb{E}^{\mathbb{Q}^{N}}\Bigg[dX_u dX_s \Bigg] = \frac{\hat{g}^{2}_{1}(T)}{T} + \frac{\min(u,s)}{t(t -\min(u,s))} + \delta(u -s) - \frac{1}{(t -\min(u,s))} du ds.  
\end{equation}
Therefore, we have that
\begin{align}\label{g_2_2_second_term}
&\frac{\nu^{2} \hat{\rho}^{2}}{T} \mathbb{E}^{\mathbb{Q}^{N}}\Bigg[ \left(\int_{0}^{T} W_{u,\perp} dW_{u} \right)^{2} \Big|W_{T} = \hat{g}_{1}(T) \sqrt{T} \Bigg] \nonumber \\
&= \frac{\nu^{2} \hat{\rho}^{2}}{T} \left( \int_{0}^{T}\int_{0}^{T} \min(u_{1},u_{2}) \mathbb{E}^{\mathbb{Q}^{N}}\Bigg[dW_{u_{1}} dW_{u_{2}}  \Big|W_{T} = \hat{g}_{1}(T) \sqrt{T} \Bigg]\right) \nonumber \\
&= \frac{\nu^{2} \hat{\rho}^{2}}{T} \int_{0}^{T}\int_{0}^{T} \min(u_{1},u_{2}) \left( \frac{\hat{g}^{2}_{1}(T)t}{t^{2}} + \frac{\min(u_{1},u_{2})}{t(t - \min(u_{1},u_{2}))} + \delta(u_1-u_2) - \frac{1}{t- \min(u_{1},u_{2})} \right) du_{1}du_{2}  \nonumber \\
&= \frac{1}{6}\nu^{2}\hat{\rho}^{2}\bigg(2 \hat{g}^{2}_{1}(T)+1\bigg)T.
\end{align}
Now, if we join \eqref{g_2_2_first_term} and \eqref{g_2_2_second_term}, we obtain that
\begin{equation}\label{sec:a13:g2_power_2}
\mathbb{E}^{\mathbb{Q}^{N}}\Bigg[\hat{g}^{2}_{1}(T)\Big|W_{T} = \hat{g}_{1}(T) \sqrt{T}\Bigg] 
=\frac{1}{12}\nu^{2}\Bigg( 3\rho^{2} \bigg(\hat{g}^{2}_{1}(T) - 1 \bigg)^{2} + 4\hat{\rho}^{2} \hat{g}^{2}_{1}(T)+ 2\hat{\rho}^{2}\Bigg) T.
\end{equation}

\section{SABR}\label{lambda_SABR_Calculations}

\subsection{Calculation of $\mathbb{E}^{\mathbb{Q}^{N}}\left[ \hat{g}_{2}(T) |\hat{g}_1(T)\right]$}
\label{appendix:b11}
From \eqref{g_s_slv}, we have that
\begin{equation*}
\mathbb{E}^{N}\Bigg[ \hat{g}_{2}(T) \Big| \hat{g}_{1}(T) \Bigg] = \alpha\partial_F C(F_0)  \mathbb{E}^{N}\Bigg[  \frac{I_{(1,1)}(T)}{\sqrt{T}} \Big| \hat{g}_{1}(T) \Bigg] + \nu \rho  \mathbb{E}^{N}\Bigg[  \frac{I_{(1,1)}(T)}{\sqrt{T}} \Big| \hat{g}_{1}(T) \Bigg]
\end{equation*}
Now, if we use \eqref{brownian_bridge_dw_dw}, we have that 
\begin{equation*}
\mathbb{E}^{N}\Bigg[\hat{g}_{2}(T) \Big| \hat{g}_{1}(T) \Bigg] = \frac{\sqrt{T}}{2} \bigg(\hat{g}^{2}_1(T) - 1\bigg) \bigg(\nu \rho + \frac{1}{2}\alpha \partial_F  C(F_0) \bigg).
\end{equation*}

\subsection{Calculation of $\mathbb{E}^{\mathbb{Q}^{N}}\big[ \hat{g}_{3}(T) |\hat{g}_1(T)\big]$}
\label{appendix:b12}
From the expression of $\hat{g}_3(T)$ in \eqref{g_s_slv}, we obtain that
\begin{equation*}
\mathbb{E}^{N}\Bigg[\hat{g}_{3}(T) \Big| \hat{g}_{1}(T) \Bigg] = A(T) + B(T) + C(T)
\end{equation*}
where
\begin{align*}
A(T) &=  \frac{\nu^2}{\sqrt{T}}\int_0^T \int_0^{u_1}\int_0^{u_2} dB_s dB_{u_2} dW_{u_1} + \alpha \nu \Big(\partial_F C(F_0)\Big) \frac{1}{\sqrt{T}}\int_{0}^{T} B_u W_u dW_u,  \\
B(T) &=  \frac{1}{2} \alpha^2 \Big(\partial^2_F C(F_0)\Big)C(F_0)\frac{1}{\sqrt{T}}\int_{0}^{T} W^2_u dW_u + \alpha  \nu \Big(\partial_F C(F_0)\Big) \frac{1}{\sqrt{T}}\int_0^{T} \int_0^{u_1} \int_0^{u_2} dB_s W_{u_2} dW_{u_1}, \\
C(T) &= \alpha^{2} \Big(\partial_{F}C(F_0)\Big)^2 \frac{1}{\sqrt{T}}I_{(1,1,1)}(T).
\end{align*}
Through a direct computation and employing the representation of the distinct iterated integrals, we get that
\begin{align*}
\mathbb{E}^{N}\Bigg[A(T)\Big|\hat{g}_1(T)\Bigg]&= \frac{1}{6}\nu^2 \rho^2 \bigg(\hat{g}^3_1(T) - 3\hat{g}_1(T)\bigg) T + \frac{1}{6}\alpha \nu \Big(\partial_F C(F_0)\Big) \bigg(2\hat{g}^3_1(T)-  3\hat{g}_1(T)\bigg) T   \\
\mathbb{E}^{N}\Bigg[B(T)\Big|\hat{g}_1(T)\Bigg]&=  \frac{1}{12}\alpha^2 \Big(\partial^2_F C(F_0)\Big)C(F_0)\bigg(2\hat{g}^3_1(T) -  3\hat{g}_1(T)\bigg)T + \frac{1}{3}\alpha \rho  \nu \Big(\partial_F C(F_0)\Big) \bigg(\hat{g}^3_1(T) -  3\hat{g}_1(T)\bigg)T \\
\mathbb{E}^{N}\Bigg[C(T)\Big|\hat{g}_1(T)\Bigg]&= \frac{1}{3}\alpha^{2} \Big(\partial_{F}C(F_0)\Big)^2  \bigg(\hat{g}^3_1(T) -  3\hat{g}_1(T) \bigg)T
\end{align*}
Rearranging the previous equalities, we obtain that
\begin{align*}
\mathbb{E}^{N}\Bigg[ \hat{g}_{3}(T) | \hat{g}_{1}(T) \Bigg] &=\frac{1}{3}\Bigg(\frac{1}{2}\nu^2 \rho^2 + \alpha \rho  \nu \partial_F C(F_0)  + \alpha^{2} \Big(\partial_{F}C(F_0)\Big)^2 \Bigg) \Bigg(\hat{g}^3_1(T) - 3\hat{g}_1(T) \Bigg)T \\
&+\frac{1}{6}\Bigg(6\alpha \nu \partial_F C(F_0) +  3\alpha^2 \Big(\partial^2_F C(F_0)\Big)C(F_0)  \Bigg) \Bigg(2\hat{g}^3_1(T) -  3\hat{g}_1(T)\Bigg)T.
\end{align*}

\subsection{Calculation of $\mathbb{E}^{\mathbb{Q}^{N}}\big[ \hat{g}^{2}_{2}(T) |\hat{g}_1(T)\big]$}
\label{appendix:b13}
From \eqref{appendix:b11} and \eqref{sec:a13:g2_power_2}, we have that
\begin{equation*}
\mathbb{E}^{N}\Bigg[ \hat{g}^{2}_{2}(T) \Big| \hat{g}_{1}(T) \Bigg] = A(T) + B(T)
\end{equation*}
where 
\begin{align*}
A(T)&= \frac{1}{4} \alpha^2 \Big(\partial_F C(F_0)\Big)^2 \Big(\hat{g}^{2}_1(T) - 1\Big)^2 T ,\\
B(T) &= \frac{\nu^{2}}{12}\bigg(3 \rho^{2}\Big(\hat{g}^{2}_{1}(T) - 1 \Big)^{2}  + 4\hat{\rho}^{2} \hat{g}^{2}_{1}(T) +2\hat{\rho}^{2}\bigg)T.
\end{align*}
Then, if we combine both expressions, we get that
\begin{equation*}
\mathbb{E}^{N}\Bigg[ \hat{g}^{2}_{2}(T) | \hat{g}_{1}(T) \Bigg] = \frac{1}{12}\Bigg(3\alpha^2 \Big(\partial_F C(F_0)\Big)^2+ \nu^{2} \rho^{2} \bigg) \Big(\hat{g}^{2}_{1}(T) - 1 \Big)^{2}  +4\nu^{2}\hat{\rho}^{2} \hat{g}^{2}_{1}(T) + 2\nu^{2}\hat{\rho}^{2}\Bigg)T.
\end{equation*}

\section{Call option expansion normal SABR}
\subsection{$\epsilon^0$-coefficient call option expansion}
In this case, we must to compute $\mathbb{E}^{\mathbb{Q}^{N}}\bigg[\Big(\hat{g}_1(T) - y\Big)^{+}\bigg]$. From the expression of $\hat{g}_1$, see \eqref{g_normal_sabr}, it is easy to show that

\begin{equation}\label{order_0_price}
	\mathbb{E}^{\mathbb{Q}^{N}}\bigg[\Big(\hat{g}_1(T) - y\Big)^{+}\bigg] = \phi(y) - y \bar{\Phi}(y) 
\end{equation}
where $\phi(\cdot)$ is the PDF for a standard normal and $\bar{\Phi}(y) = 1 - \Phi(y)$ with $\Phi(\cdot)$ the CDF for a standard normal.

\subsection{$\epsilon$-coefficient call option expansion}
We have to compute $\mathbb{E}^{\mathbb{Q}^{N}}\bigg[I\Big(\hat{g}_{1}(T) > y\Big) \hat{g}_{2}(T) \bigg]$. From \ref{appendix:a11}, we have that
\begin{align}\label{order_1_price}
\mathbb{E}^{\mathbb{Q}^{N}}\Bigg[ I\Big(\hat{g}_{1}(T) > y\Big) \hat{g}_2(T)\Bigg]&=\mathbb{E}^{\mathbb{Q}^{N}}\Bigg[ I\Big(\hat{g}_{1}(T) > y\Big) \mathbb{E}^{\mathbb{Q}^{N}}\Big[\hat{g}_2(T) | \hat{g}_1(T) \Big] \Bigg]  \nonumber\\
&= \frac{\rho \nu \sqrt{T}}{2}  y \phi(y).
\end{align}

\subsection{$\epsilon^2$-coefficient call option expansion}
As we did before, we have that 
\begin{align}\label{order_3_term_1}
\mathbb{E}^{\mathbb{Q}^{N}}\Bigg[ I\Big(\hat{g}_1(T) > y\Big) \hat{g}_3(T) \Bigg] &= \mathbb{E}^{\mathbb{Q}^{N}}\Bigg[I\Big(\hat{g}_1(T) > y\Big) \frac{\nu^{2}}{12}\bigg(\rho^{2}  \Big( 2 \hat{g}^{3}_{1}(T) - \hat{g}_{1}(T)\Big)  + 3\hat{\rho}^{2} \hat{g}_{1}(T) - 3 \hat{g}_{1}(T)  \bigg)T\Bigg] \nonumber \\
&= \frac{1}{12}\nu^{2}\phi(y) \bigg(\rho^{2}\left(2y^{2} + 1\right) + 3\hat{\rho}^{2} - 3 \bigg) T,
\end{align}
and
\begin{align}\label{order_3_term_2}
\mathbb{E}^{\mathbb{Q}^{N}}\Bigg[ \frac{1}{2}\delta\Big(\hat{g}_1(T) - y\Big) \hat{g}^{2}_{2}(T)\Bigg] &= \frac{1}{24}\nu^{2}\phi(y)\bigg( 3 \rho^{2}\Big(y^{2} - 1 \Big)^{2} + 4\hat{\rho}^{2} y + 2\hat{\rho}^{2}\bigg) T.
\end{align}

\section{Quadratic call option local volatility}
\subsection{$\epsilon^0$-coefficient quadratic call option expansion} \label{epsilon_0_quadratic_call}
We must to compute $\mathbb{E}^{\mathbb{Q}^{N}}\Bigg[\Big((\hat{g}_1 - y)^{+}\Big)^{2}\Bigg]$. From the expression of $\hat{g}_1$, see \eqref{g_normal_sabr}, it   is easy to show that
\begin{equation}\label{qc_order_0_price_lv}
	\mathbb{E}^{\mathbb{Q}^{N}}\Bigg[\Big((\hat{g}_1(T) - y)^{+}\Big)^{2}\Bigg] =  \Big(1+y^{2}\Big) \bar{\Phi}(y) - y \phi(y). 
\end{equation}
\subsection{$\epsilon$-coefficient quadratic call option expansion}\label{epsilon_1_quadratic_call}
To calculate this term, we have to calculate $\mathbb{E}^{\mathbb{Q}^{N}}\Bigg[2 \Big(\hat{g}_1(T) - y\Big)^{+} \hat{g}_2(T)\Bigg]$. It is easy to show using integration by parts that

\begin{equation}\label{qc_order_1_price_lv}
\mathbb{E}^{\mathbb{Q}^{N}}\Bigg[2\Big(\hat{g}_1(T) - y\Big)^{+} \hat{g}_{2}(T) \Bigg] = \Big(\partial_{F}\sigma(F_0)\Big) \sqrt{T} \phi(y).
\end{equation}

\subsection{$\epsilon^2$-coefficient quadratic call option expansion}\label{epsilon_2_quadratic_call}
The last term that we will calculate is 

\begin{equation*}
\mathbb{E}^{\mathbb{Q}^{N}}\Bigg[2 \Big(\hat{g}_1(T) - y\Big)^{+}  \hat{g}_3(T) + I\Big(\hat{g}_1(T) > y\Big) \hat{g}^{2}_2(T)\Bigg] = \mathbb{E}^{\mathbb{Q}^{N}}\Bigg[A(T) + B(T)\Bigg]
\end{equation*}
where
\begin{align*}
A(T) &= 2 \Big(\hat{g}_1(T) - y\Big)^{+} \hat{g}_3(T), \\
B(T) &= I\Big(\hat{g}_1(T) > y\Big) \hat{g}^{2}_2(T).
\end{align*}
From \eqref{coeff_lv_expansion}, we have that 
\begin{equation*}
B(T) =  \frac{1}{4}I\Big(\hat{g}_1(T) > y\Big) \Big(\partial_F\sigma(F_0)\Big)^2  \Big(\hat{g}^{2}_1(T) - 1\Big)^2 T.
\end{equation*}
Also, we get that
\begin{align*}
A(T) =  \Big(\hat{g}_1(T) - y\Big)^{+} \Bigg(\frac{1}{3}\bigg( &\Big(\partial^{2}_F \sigma(F_0)\Big) \sigma(F_0) + \Big(\partial_F \sigma(F_0)\Big)^{2}\bigg) \bigg(\hat{g}^{3}_1(T) - 3\hat{g}_1(T)\bigg)T\\
&  + \Big(\partial^{2}_F \sigma(F_0)\Big) \sigma(F_0)\bigg(\hat{g}_1(T)T - \frac{1}{\sqrt{T}}\int_{0}^{T} W_u du  \bigg) \Bigg).
\end{align*} 
Therefore, we have that
\begin{align*}
\mathbb{E}^{\mathbb{Q}^{N}}\Bigg[2 \Big(\hat{g}_1(T) - y\Big)^{+}  \hat{g}_3(T)\Bigg] &=\frac{1}{6}\Bigg( 2\bigg( \Big(\partial^{2}_F \sigma(F_0)\Big) \sigma(F_0) + \Big(\partial_F \sigma(F_0)\Big)^{2}\bigg) y \phi(y) + 3\Big(\partial^{2}_F\sigma(F_0)\Big)\sigma(F_0)\bar{\Phi}(y) \Bigg) T ,  \nonumber \\
\mathbb{E}^{\mathbb{Q}^{N}}\Bigg[ I\Big(\hat{g}_1(T) > y\Big) \hat{g}^{2}_2(T) \Bigg] &= \frac{1}{4}\Big(\partial_F \sigma(F_0)\Big)^{2}  \bigg(\Big(y^3+y\Big)\phi(y)+ 2\bar{\Phi}(y)\bigg)T.
\end{align*}

\section{Quadratic call option SABR}
\subsection{$\epsilon^0$-coefficient quadratic call option expansion}\label{appendix:e51}
We do the same calculation that in \eqref{qc_order_0_price_lv}.

\subsection{$\epsilon$-coefficient quadratic call option expansion}\label{appendix:e52}
We have to compute $\mathbb{E}^{\mathbb{Q}^{N}}\Bigg[2\Big(\hat{g}_{1}(T) - y\Big)^{+} \hat{g}_{2}(T) \Bigg]$. Then, 
\begin{equation*}
\mathbb{E}^{\mathbb{Q}^{N}}\Bigg[2\Big(\hat{g}_{1}(T) - y\Big)^{+} \hat{g}_{2}(T) \Bigg] = \mathbb{E}^{\mathbb{Q}^{N}}\Bigg[2\Big(\hat{g}_{1}(T) - y\Big)^{+}  \mathbb{E}^{\mathbb{Q}^{N}}\Big[ \hat{g}_{2}(T)| \hat{g}_{1}(T)\Big]\Bigg].
\end{equation*}
Using \eqref{appendix:b11}, we have that
\begin{align*}
\mathbb{E}^{\mathbb{Q}^{N}}\Bigg[2\Big(\hat{g}_{1}(T) - y\Big)^{+}  \mathbb{E}^{\mathbb{Q}^{N}}\Big[ \hat{g}_{2}(T)| \hat{g}_{1}(T)\Big]\Bigg] &= \sqrt{T} \bigg(\nu \rho + \frac{1}{2}\alpha \Big(\partial_F  C(F_0)\Big) \bigg) \mathbb{E}^{\mathbb{Q}^{N}}\bigg[\Big(\hat{g}_{1}(T) - y\Big)^{+} \Big(\hat{g}^{2}_{1}(T) - 1\Big)  \bigg] \\
&= \sqrt{T} \bigg(\nu \rho + \frac{1}{2}\alpha \Big(\partial_F  C(F_0)\Big) \bigg) \int_{y}^{\infty} (s - 1)(s^2 - 1) \phi(s) ds \nonumber \\
& = \left(\nu \rho + \frac{1}{2}\alpha \Big(\partial_F  C(F_0)\Big) \right)  \phi(y) \sqrt{T}
\end{align*}

\subsection{$\epsilon^{2}$-coefficient quadratic call option expansion}\label{appendix:e53}
As we did in \eqref{epsilon_2_quadratic_call}, we must calculate
\begin{equation*}
\mathbb{E}^{\mathbb{Q}^{N}}\Bigg[  2 \Big(\hat{g}_1(T) - y\Big)^{+}  \hat{g}_3(T) + I\Big(\hat{g}_1(T) > y\Big) \hat{g}^{2}_2(T)  \Bigg]= A(T) + B(T)
\end{equation*}
where 
\begin{align*}
A(T) &= 2 \mathbb{E}^{\mathbb{Q}^{N}}\Bigg[\Big(\hat{g}_1(T) - y\Big)^{+}  \hat{g}_3(T) \Bigg],\\
B(T)&= \mathbb{E}^{\mathbb{Q}^{N}}\Bigg[I\Big(\hat{g}_1(T) > y\Big) \hat{g}^{2}_2(T) \Bigg].
\end{align*}
From \eqref{g_s_slv} and using \eqref{appendix:b11},\eqref{appendix:b12} and \eqref{appendix:b13}, we obtain that
\begin{align*}
A(T) &=  \frac{2}{3}\bigg(\frac{1}{2}\nu^2 \rho^2 + \alpha \rho  \nu \Big(\partial_F C(F_0)\Big)  + \alpha^{2} \Big(\partial_{F}C(F_0)\Big)^2 \bigg)y\phi(y)T \\
&= \bigg(\alpha \nu \Big(\partial_F C(F_0)\Big) +  \frac{1}{2} \alpha^2 \Big(\partial^2_F C(F_0)\Big)C(F_0)\bigg) \bigg( \frac{2}{3}y\phi(y) + \bar{\Phi}(y) \bigg)T,
\end{align*}
and
\begin{equation*}
B(T)= \frac{1}{12}\Bigg(\bigg(3\alpha^2 \Big(\partial_F C(F_0)\Big)^2+ \nu^{2} \rho^{2} \bigg) \bigg(2\bar{\Phi}(y)+\Big(y^3+y\Big)\phi(y) \bigg)  + 4 \nu^2 \hat{\sigma}^2y \phi(y)  + 6 \nu^2 \hat{\sigma}^2\bar{\Phi}(y)  \Bigg)T .
\end{equation*}

\bibliography{references/references, references/references-books, references/references-own, references/references-online}

\end{document}